\documentclass[11pt,a4paper]{article}
\usepackage{amsthm}
\RequirePackage{fullpage}
\usepackage{booktabs}
\RequirePackage{enumerate} 
\RequirePackage[latin1]{inputenc}
\RequirePackage{amsmath}
\RequirePackage{amssymb}

\usepackage{xcolor}
\usepackage{hyperref}
\hypersetup{
    colorlinks=true,                
    linkcolor= blue,                
    pdfborder={0 0 1},
    linkbordercolor={1 0 0},
    citebordercolor=green
}
\usepackage{algorithmicx}
\usepackage{algorithm}
\usepackage[noend]{algpseudocode}

\algdef{SE}[DOWHILE]{Do}{doWhile}{\algorithmicdo}[1]{\algorithmicwhile\ #1}

\makeatletter
\newcommand{\leqnos}{\tagsleft@true\let\veqno\@@leqno}
\newcommand{\reqnos}{\tagsleft@false\let\veqno\@@eqno}
\reqnos
\makeatother

\usepackage{amsmath,amsfonts}
\usepackage{tcolorbox}
\usepackage{caption}
\usepackage{subcaption}
\usepackage{epsfig}

\newcommand{\CCase}[1]{\noindent {\bf Case #1:}}

\newtheorem{theorem}{Theorem}[section]

\newtheorem{lemma}[theorem]{Lemma}

\newtheorem{observation}[theorem]{Observation}

\newtheorem{definition}[theorem]{Definition}
\newtheorem{conjecture}[theorem]{Conjecture}

\newtheorem{example}[theorem]{Example}
\newtheorem{open-problem}[theorem]{Open Problem}
\newcommand{\beq}{\begin{equation}}
\newcommand{\eeq}{\end{equation}}
\makeatletter
\newcommand\newtag[2]{#1\def\@currentlabel{#1}\label{#2}}
\makeatother

\newcommand{\hidetext}[1]{}
\def\P{\mathbb{P}}

\begin{document}


\title{Toward a Dichotomy for Approximation of H-coloring\thanks{An extended abstract of this work appeared in ICALP 2019~\cite{RafieyRS19}}}

 \author{Akbar Rafiey \thanks{Department of Computing Science, Simon Fraser University; arafiey@sfu.ca. Supported by NSERC}   
      \and 
Arash Rafiey \thanks{Department of Math and Computer Science, Indiana State University; arash.rafiey@indstate.edu. Department of Computer Science, Simon Fraser University; arashr@sfu.ca. Supported by NSF 1751765.}  
\and 
Thiago Santos \thanks {Department of Math and Computer Science, Indiana State University; tsantos2@sycamores.indstate.edu}
}

\date{}
\maketitle

\begin{abstract}
Given two (di)graphs $G$, $H$ and a cost function $c:V(G)\times V(H) \to \mathbb{Q}_{\geq 0}\cup\{+\infty\}$, in the minimum cost homomorphism problem, MinHOM($H$), we are interested in finding a homomorphism $f:V(G)\to V(H)$ (a.k.a $H$-coloring) that minimizes $\sum\limits_{v\in V(G)}c(v,f(v))$. The complexity of \emph{exact minimization} of this problem is well understood~\cite{hell2012dichotomy}, and the class of digraphs $H$, for which the MinHOM($H$) is polynomial time solvable is a small subset of all digraphs. 

In this paper, we consider the approximation of MinHOM within a constant factor. In terms of digraphs, MinHOM($H$) is not approximable if $H$ contains a \emph{digraph asteroidal triple (DAT)}. We take a major step toward a dichotomy classification of approximable cases. We give a dichotomy classification for approximating the MinHOM($H$) when $H$ is a graph (i.e. symmetric digraph). For digraphs, we provide constant factor approximation algorithms for two important classes of digraphs, namely bi-arc digraphs (digraphs with a \emph{conservative semilattice polymorphism} or \emph{min-ordering}), and
$k$-arc digraphs (digraphs with an \emph{extended min-ordering}).  Specifically, we show that:
\begin{itemize} 
 \item \textbf{Dichotomy for Graphs:} MinHOM($H$) has a $2|V(H)|$-approximation algorithm if graph $H$ admits a \emph{conservative majority polymorphims} (i.e. $H$ is a \emph{bi-arc graph}), otherwise, it is inapproximable;
 \item MinHOM($H$) has a $|V(H)|^2$-approximation algorithm if $H$ is a bi-arc digraph;
    \item MinHOM($H$) has a $|V(H)|^2$-approximation algorithm if $H$ is a $k$-arc digraph.
\end{itemize}

Our constant factors depend on the size of $H$. However, the implementation of our algorithms provides 
a much better approximation ratio. 
It leaves open to investigate a classification of digraphs $H$, where MinHOM($H$) admits a constant factor approximation
algorithm that is independent of $|V(H)|$.
\end{abstract}
\newpage
\section{Introduction}

For a digraph $D$, let $V(D)$ denote the vertex set of $D$, and let $A(D)$ denote
the arcs of $D$. We denote the number of vertices of $D$ by  $|D|$. Instead of $(u,v) \in A(D)$, we use the shorthand
$uv \in A(D)$ or simply $uv\in D$.

A graph $G$ is a symmetric digraph, that is, $xy \in A(G)$ if and only if $yx \in A(G)$. An edge is just a symmetric arc. 

A {\em homomorphism} of a digraph $D$ to a digraph $H$ (a.k.a $H$-\textsc{Coloring}) is a mapping $f : V(D) \rightarrow V(H)$ such that for each arc $xy$ of $D$, $f(x)f(y)$ is an arc of $H$. We say the map $f$ does not satisfy arc $xy$, if $f(x)f(y)$ is not an arc of $H$. The homomorphism problem for a fixed target digraph $H$, HOM($H$), takes a digraph $D$ as input and asks whether there is a homomorphism from $D$ to $H$. 
Therefore, by fixing the digraph $H$ we obtain a class of problems, one for each digraph $H$. 
For example, HOM($H$), when $H$ is an edge, is exactly the problem of determining whether the input graph $G$ is bipartite (i.e., the \textsc{2-Coloring} problem). Similarly, if $V(H)=\{u,v,x\}, A(H)=\{uv,vu,vx,xv,ux,xu\}$, then HOM($H$) is exactly the classical \textsc{3-Coloring} problem. More generally, if $H$ is a clique on $k$ vertices, then HOM($H$) is the $k$-\textsc{Coloring} problem. The HOM($H$) problem can be considered within a more general framework, the constraint satisfaction problem (CSP). In the CSP associated with a finite relational structure $\mathbb{H}$, the question is whether there exists a homomorphism of a given finite relational structure to $\mathbb{H}$. Thus, the $H$-\textsc{Coloring} problem is a particular case of the CSP in which the involved relational structures are digraphs. A celebrated result due to Hell and Nesetril~\cite{hell1990complexity}, states that, for graph $H$, HOM($H$) is in P if $H$ is bipartite or contains a looped vertex, and that it is NP-complete for all other graphs $H$. See~\cite{bulatov2005h} for an algebraic proof of the same result, and~\cite{bulatov17,zhuk17} for a dichotomy for CSP($\mathbb{H}$).

There are several natural optimization versions of the HOM($H$) problem. One is to find a mapping $f : V(D) \rightarrow V(H)$ that maximizes (minimizes) the number of satisfied (unsatisfied) arcs in $D$. This problem is known under the name of \textsc{Max 2-CSP} (\textsc{Min 2-Csp}). For example, the most basic \textsc{Boolean Max 2-CSP}
problem is \textsc{Max Cut} where the target graph $H$ is an edge. This line of research has received a lot of attention in the literature and there are very strong results concerning various aspects of approximability of \textsc{Max 2-CSP} and \textsc{Min 2-CSP}~\cite{austrin2010towards,goemans1995improved,haastad2001some,khot2007optimal,lewin2002improved}. See~\cite{makarychev2017approximation} for a recent survey on this and approximation of \textsc{Max $k$-CSP} and \textsc{Min $k$-CSP}. We consider another natural optimization version of the HOM($H$) problem, i.e., we are not only interested in the existence of a homomorphism, but want to find the "best homomorphism". The {\em minimum cost homomorphism problem} to $H$, denoted by MinHOM($H$), for a given input digraph $D$, and a cost function $c(x,i), x \in V(D), i \in V(H)$,
seeks a homomorphism $f$ of $D$ to $H$ that minimizes the total cost $\sum\limits_{ x \in V(D)} c(x,f(x))$. The cost function $c$ can take non-negative rational values and positive infinity, that is $c:V(D)\times V(H)\to \mathbb{Q}_{\geq 0}\cup \{+\infty\}$. MinHOM was introduced in~\cite{gutin2006level}, where it was motivated by a real-world problem in defence logistics. MinHOM problem offers a natural and practical way to model and generalizes many optimization problems.

\begin{example}[\textsc{(Weighted) Minimum Vertex Cover}]
This problem can be seen as MinHOM($H$) where $V(H)=\{0,1\}, E(H)=\{11,01\}$ 
and $c(u,0)=0$, $c(u,1)> 0$ for every $ u \in V(G)$. Note that $G$ and $H$ are graphs in this example.
\end{example}

\begin{example}[\textsc{List Homomorphism (LHOM)}]
LHOM($H$), seeks, for a given input digraph
$D$ and lists $L(x) \subseteq V(H), x \in V(D)$, a homomorphism $f$ from $D$ to $H$ such that $f(x) \in L(x)$ for all $x \in V(D)$. This is equivalent to MinHOM($H$) with $c(u,i)=0$ if $i \in L(u)$, otherwise $c(u,i)=+\infty$. This problem is also known as \textsc{List $H$-Coloring} and its complexity is fully understood due to series of results~\cite{barto2011dichotomy,bulatov2003tractable,bulatov2011complexity,bulatov2016conservative,feder2003bi,hell2011dichotomy}.
\end{example}

The MinHOM problem generalizes many other problems such as \textsc{(Weighted) Min Ones}~\cite{arora1997hardness,creignou2001complexity,khanna2001approximability}, \textsc{Min Sol}~\cite{jonsson2008introduction,uppman2013complexity}, a large class of bounded
integer linear programs, retraction problems~\cite{feder2010retractions}, \textsc{Minimum Sum Coloring}~\cite{bar1998chromatic,giaro200227,kubicka1989introduction}, and various optimum cost chromatic partition
problems~\cite{halldorsson2001minimizing,jansen2000approximation,jiang1999coloring,kroon1997}. 

A special case of MinHOM problem is where the cost function $c$ is choosen from a fixed set $\Delta$. This problem is denoted by MinHOM($H,\Delta$)~\cite{cohen2017binarisation,uppman2013complexity,uppman2014computational}. The \textsc{Valued Constrained Satisfaction Problems} (VCSPs) is a generalization of this special case of the MinHOM problem. An instance of the VCSP is given by a collection of variables that must be assigned labels from a given domain with the goal to minimize the objective function that is given by 
the sum of cost functions, each depending on some subset of the variables~\cite{cohen2006complexity}. Interestingly, a recent work by Cohen \emph{et al.}~\cite{cohen2017binarisation} proved that VCSPs over a fixed valued constraint language are polynomial-time equivalent to MinHOM($H,\Delta$) over a fixed digraph and a proper choice of $\Delta$.

\paragraph{\textbf{Exact Minimization}}The complexity of \emph{exact minimization} of MinHOM($H$) was studied in a series of papers, and complete complexity classifications were given in~\cite{mincostungraph} for undirected graphs, in~\cite{hell2012dichotomy} for digraphs, and in~\cite{takhanov2007dichotomy} for more general structures. Certain minimum cost homomorphism problems
have polynomial-time algorithms~\cite{mincostungraph,yeo,gutin2006level,hell2012dichotomy}, but most are NP-hard. We remark that, the complexity of \emph{exact minimization} of VCSPs is well understood~\cite{kolmogorov2017complexity,thapper2016complexity}.

\paragraph{\textbf{Approximation}}For a minimization
problem, an \emph{$\alpha$-approximation} algorithm is a (randomized) polynomial-time algorithm that finds an approximate solution of cost at most $\alpha$ times the minimum cost. A constant ratio approximation algorithm is an
$\alpha$-approximation algorithm for some constant $\alpha$. The approximability of MinHOM is fairly understood when we restrict the cost function to a fixed set $\Delta$, and further, we restrict it to take only finite values (not $\infty$). This setting is a special case of finite VCSPs, and there are strong approximation results on finite VCSPs. For finite VCSPs, Raghavendra~\cite{raghavendra2008optimal} showed how to use the basic SDP relaxation to obtain a constant approximation. Moreover, he proved that the approximation ratio cannot be improved under \textsc{Unique Game Conjecture} (UGC). This constant is not explicit, but there is an algorithm that can compute it with any given accuracy in doubly exponential time. In another line of research, the power of so-called \emph{basic linear program (BLP)} concerning constant factor approximation of finite VCSPs has been recently studied in~\cite{dalmau2018towards,ene2013local}. However, the approximability of VCSPs for constraint languages that are not finite-valued remains poorly understood, and~\cite{arash-esa,jonsson2008introduction} are the only results on approximation of VCSP for languages that have cost functions that can take infinite values. 

Hell \emph{et al.,}~\cite{arash-esa} proved a dichotomy for approximating MinHOM($H$) when $H$ is a bipartite graph by transforming the MinHOM($H$) to a linear program, and rounding the fractional values to get a homomorphism to $H$.

\begin{theorem}[Dichotomy for bipartite graphs~\cite{arash-esa}]
For a fixed bipartite graph $H$, MinHOM($H$) admits a constant factor approximation algorithm if $H$ admits a \emph{min-ordering} (complement of $H$ is a \emph{circular arc} graph), otherwise MinHOM($H$) is not approximable unless P = NP.
\end{theorem}

 Beyond this, there is no result concerning the approximation of MinHOM($H$). We go beyond the bipartite case and present a constant factor approximation algorithm for \emph{bi-arc} graphs (graphs with a \emph{conservative majority} polymorphism). Designing an approximation algorithm for MinHOM($H$) when $H$ is a digraph is  much more complex than when $H$ is a graph. We improve the state of the art by providing constant factor approximation algorithms for MinHOM($H$) where $H$ belongs to two important classes of digraphs, namely \emph{bi-arc} digraphs (digraphs with a \emph{conservative semilattice} polymorphism a.k.a min-ordering), and \emph{k-arc} digraphs (digraphs with a $k$-min-ordering). To do so, we introduce new LPs that reflect the structural properties of the target (di)graph $H$ as well as new methods to round the fractional solutions and obtain homomorphisms to $H$. We will show our randomized rounding procedure can be de-randomized, and hence, we get a deterministic polynomial algorithm. Furthermore, we argue that our techniques can be used towards finding a dichotomy for the approximation of MinHOM($H$).

\subsection{Our Contributions} 
We say that a problem is not approximable if there is no polynomial-time approximation algorithm with a multiplicative guarantee unless P = NP.
Most of the minimum cost homomorphism problems are NP-hard, therefore we investigate the approximation of MinHOM($H$).
\begin{tcolorbox}[colback=white]
	{\sc Approximating Minimum Cost Homomorphism to Digraph $H$:}
	\begin{quote}
	\emph{Input:} A digraph $D$ and a vertex-mapping costs $c(x,u), x \in V(D), u \in V(H)$,\\
	\emph{Output:} A homomorphism $f$ of $D$ to $H$ with the total cost of $\sum\limits_{\small x \in V(D)} c(x,f(x))\leq \alpha\cdot OPT$, where $\alpha$ is a constant.
\end{quote}
\end{tcolorbox}
Here, $OPT$ denotes the cost of a minimum cost homomorphism of $D$ to $H$. 
Moreover, we assume that the size of $H$ is constant. Recall that we approximate the cost over real homomorphisms, rather than approximating the maximum weight of satisfied constraints, as in, say, \textsc{Max CSP}. One can show that if LHOM($H$) is not polynomial-time solvable then there is no approximation algorithm for MinHOM($H$)~\cite{arash-esa, approximation-2011}.  
\begin{observation}
If LHOM($H$) is not polynomial-time solvable, then there is no approximation algorithm for MinHOM($H$).
\end{observation}
The complexity of the LHOM problems for graphs, digraphs, and relational structures (with arity two and higher)  have been classified in \cite{feder2003bi,hell2011dichotomy,bulatov2011complexity} respectively. LHOM($H$) is polynomial-time solvable if the digraph $H$ does not contain a \emph{digraph asteroidal triple (DAT)}\footnote{The definition of DAT (Definition ~\ref{def-DAT}) is rather technical and it is not necessary to fully understand in this paper.}  as an induced sub-digraph, and NP-complete when $H$ contains a DAT~\cite{hell2011dichotomy}. 

MinHOM($H$) is polynomial-time solvable when digraph $H$ admits a $k$-min-max-ordering, a subclass of DAT-free digraphs, and otherwise,  NP-complete ~\cite{hell2012dichotomy, monoton-proper}. 

In this paper, we take an important step towards closing the gap between DAT-free digraphs and the ones that admit a $k$-min-max-ordering. First, we consider digraphs that admit a min-ordering. Digraphs that admit a min-ordering have been studied under the name of \emph{bi-arc} digraphs~\cite{hell2016bi} and \emph{signed interval} digraphs~\cite{Ross,hell2020min}. Deciding if digraph $H$ has a min-ordering and finding a min-ordering of $H$ is in P~\cite{hell2016bi}. We provide a constant factor approximation algorithm for MinHOM($H$) where $H$ admits a min-ordering. 
\begin{theorem}[Digraphs with a min-ordering]\label{bi-arc-digraph-theorem}
If digraph $H$ admits a min-ordering, then MinHOM($H$) has a constant factor approximation algorithm. 
\end{theorem}

Sections \ref{min-to-max}, \ref{sec-Approx-ALG} are dedicated to the proof of Theorem \ref{bi-arc-digraph-theorem}.  
In section \ref{extended-min-section}, we turn our attention to digraphs with $k$-min-orderings, for integer $k>1$. They are also called digraphs with \emph{extended $X$-underbar}~\cite{bagan,woginger,gary}. It was shown in~\cite{woginger} that if $H$ has the $X$-underbar property, then the  HOM($H$) problem is polynomial-time solvable. In Lemmas~\ref{List-extended-min-ordering-1} and \ref{List-extended-min-ordering-2}, we show that if $H$ admits a $k$-min-ordering, then $H$ is a DAT-free digraph, and provide a simple proof that LHOM($H$) is polynomial-time solvable. Note that in general if $H$ is a DAT-free digraph, then LHOM($H$) is polynomial-time solvable \cite{hell2011dichotomy}. Finally, we have the following theorem. 

\begin{theorem}[Digraphs with a $k$-min-ordering]\label{k-arc-theorem}
If digraph $H$ admits a $k$-min-ordering for some integer $k>1$, then MinHOM($H$) has a constant factor approximation algorithm. 
\end{theorem}

Considering graphs, Feder \emph{et al.,}~\cite{feder2003bi} proved that LHOM($H$) is polynomial-time solvable if $H$ is a \emph{bi-arc} graph, and is NP-complete otherwise. In the same paper (\cite{feder2003bi}, Theorem 4.2), they mentioned that a graph $H$ is a bi-arc graph if and only if it admits a conservative majority polymorphism. In Section \ref{sec-majority}, we show that the same dichotomy classification holds in terms of approximation. 

\begin{theorem}[Dichotomy for graphs]
There exists a constant factor approximation algorithm for MinHOM(H) if $H$ is a bi-arc graph, otherwise, MinHOM(H) is inapproximable unless P = NP. 
\end{theorem}

By combining the approach for obtaining the dichotomy in the graph case, together with the idea of getting an approximation algorithm for digraphs admitting a min-ordering, we might be able to achieve a constant factor approximation algorithm for MinHOM($H$) when $H$ is DAT-free. 
\begin{conjecture}
MinHOM($H$) admits a constant factor approximation algorithm when $H$ is a DAT-free digraph, otherwise, MinHOM($H$) is not approximable unless P = NP. 
\end{conjecture}

Our constant factors depend on the size of $H$. However, the implementation of the LP and the  ILP would yield a small integrality gap (Section \ref{sec-experiment}). This indicates perhaps a better analysis of the performance of our algorithm is possible. 

\begin{open-problem}
For which digraphs MinHOM($H$) is approximable within a constant factor independent of size of $H$?
\end{open-problem}

\section{Definitions and Preliminaries}

Complexity and approximation of the minimum cost homomorphism problems, and in general the constraint satisfaction problems, are often studied under the existence of \emph{polymorphisms}~\cite{barto2017polymorphisms}.
A polymorphism of $H$ of arity $k$ is a mapping $f$ from the set of $k$-tuples over $V(H)$ to $V(H)$ such that if $x_iy_i \in A(H)$ for $i = 1, 2, \dots, k$, then $f(x_1,x_2,\dots,x_k)f(y_1,y_2,\dots,y_k) \in A(H)$. If $f$ is a polymorphism of $H$ we also say that $H$ admits $f$. A polymorphism $f$ is \emph{idempotent} if it satisfies $f(x,x,\dots,x) = x$ for all $x\in V(H)$, and is \emph{conservative} if $f(x_1,x_2,\dots,x_k)\in \{x_1,x_2,\dots,x_k\}$.
A conservative \emph{semilattice} polymorphism is a conservative binary polymorphism that is associative, idempotent, commutative.
A conservative \emph{majority} polymorphism $\mu$ of $H$ is a conservative ternary polymorphism such that  $\mu(x,x,y)=\mu(x,y,x)=\mu(y,x,x)=x$ for all $x,y\in V(H)$. 

A conservative semilattice polymorphism of $H$ naturally defines a binary relation $x \leq y$ on the vertices of $H$ by $x \leq y$ if and only if $f(x,y)=x$; by associative, the relation $\leq$ is a linear order on $V(H)$, which we call a {\em min-ordering} of $H$. 
\begin{definition}
The ordering $v_1 < v_2 < \dots < v_n$ of $V(H)$ is a
\begin{itemize}
\item[--]  \emph{min-ordering} if and only if $uv\in A(H), u'v'\in A(H)$ and $u<u', v'<v$ implies that $uv' \in A(H)$;
\item [--] \emph{max-ordering} if and only if $uv \in A(H), u'v' \in A(H)$ and $u < u', v' < v$ implies that $u'v \in A(H)$;
\item [--] \emph{min-max-ordering} if and only if $uv \in A(H), u'v' \in A(H)$ and $u < u', v' < v$ implies that $uv',u'v \in A(H)$.
\end{itemize}
\end{definition}

For a bipartite graph $H=(B,W)$ let $\overrightarrow{H}$ be the digraph obtained by orienting all the edges of $H$ from $B$ to $W$. If $\overrightarrow{H}$ admits a min-ordering then we say $H$ admits a min-ordering. This is equivalent to the following definition of min-ordering  for bipartite graphs. We say $H$ admits a min-ordering if there is an ordering $x_1,x_2,\dots,x_p$ of the vertices in $B$ and an ordering 
$y_1,y_2,\dots,y_q$ of the vertices in $W$, so that whenever $x_iy_j \in E$ and $x_ry_s$ (with $i<r$ and $s<j$) are edges of $H$ then 
$x_ry_s$ is an edge of $H$. Min-max-ordering for bipartite graphs is defined similarly. 

It is worth mentioning that, a bipartite graph $H$ admits a conservative majority if and only if it admits a min-ordering~\cite{arash-esa}. Moreover, the complement of $H$ is a circular arc graph with clique cover two \cite{feder2003bi}. 

\begin{definition}
Let $H=(V,E)$ be a digraph that admits a homomorphism $f : V(H) \rightarrow  \overrightarrow{C_k}$ (here  $\overrightarrow{C_k}$ is the induced directed cycle on $\{0,1,2,\dots,k-1\}$( i.e., arc set $\{(01,12,23,...,(k-2)(k-1),(k-1)0\}$). Let $V_i=f^{-1}(i)$, $0 \le i \le k-1$. 
\begin{itemize}
\item [--]A $k$-{\em min-ordering} of $H$ is a linear ordering $<$ of the vertices of $H$, so that $<$ is a min-ordering on the subgraph induced by any
two circularly consecutive $V_i, V_{i+1}$ (subscript addition modulo $k$).
\item [--]A $k$-{\em min-max-ordering} of $H$ is a linear ordering $<$ of the vertices of $H$, so
that $<$ is a min-max-ordering on the subgraph induced by any two circularly consecutive $V_i, V_{i+1}$ (subscript addition modulo $k$).
\end{itemize}
\end{definition}

We close this section by giving a formal definition of a digraph asteroidal triple (DAT). The definition is rather technical and it is not necessary to fully understand it in this paper. We
give a brief discussion on DAT for the sake of completeness.

\begin{definition}[Invertible Pair]
Let $H$ be a digraph. Define $\widehat{H^k}$ to be the digraph with the vertex set $\{ (a_1,a_2,\dots,a_k) \mid a_i \in V(H), 1 \le i \le k\}$ and the arc set 

$A(\widehat{H^k})=\{ 
(a_1,\dots,a_k)(b_1,\dots,b_k) \mid a_ib_i \in A(H),  1 \le i \le k, 
 a_1b_j \not\in A(H) \forall j,  2 \le j \le k  \} \cup  \{(a_1,\dots,a_k)(b_1,\dots,b_k) \mid b_ia_i \in A(H),  1 \le i \le k, 
 b_ja_1 \not\in A(H) \forall j,  2 \le j \le k  \}$.

We say $(x,y)$ is an invertible pair if $(x,y),(y,x)$ belong to the same strong component of $\widehat{H^2}$. 
\end{definition}
\begin{definition}[DAT]\label{def-DAT}
A digraph asteroidal triple of $H$ is an induced sub-digraph of $\widehat{H^3}$, on three directed paths $P_1,P_2,P_3$ where $P_1$ goes from $(a,b,c)$ to $(\alpha,\beta,$ $\beta)$,
$P_2$ goes from $(b,a,c)$ to $(\alpha,\beta,\beta)$, and $P_3$ goes from $(c,a,b)$ to $(\alpha,\beta,\beta)$ and $(\alpha,\beta)$ is an invertible pair.
\end{definition}
If $H$ contains a DAT then all three pairs $(a,b),(b,c),(c,a)$ are invertible. Note that an invertible pair is an obstruction to existence of min-orderings \cite{feder2003bi, arash-esa}. Moreover, $H$ does not admit a conservative majority polymorphism $g$ because of the directed path $P_1$, $g(a,b,c) \ne a$, and because of $P_2$, $g(a,b,c) \ne b$, and finally because of $P_3$, $g(a,b,c) \ne c$. Therefore, the value of $g(a,b,c)$ can not be any of the $a,b,c$ \cite{hell2011dichotomy}. 

\section{LP for Digraphs with a min-max-ordering}\label{sec-minmax}
 
Before presenting the LP, we give a procedure to modify the lists associated to the vertices of $D$.
To  each vertex  $x \in D$, we associate a list $L(x)$ that initially contains $V(H)$. Think of $L(x)$ as the set of possible images for $x$ in a homomorphism from $D$ to $H$. Apply the \emph{arc consistency} procedure as follows.  Take an arbitrary arc $xy \in A(D)$ ($yx \in A(D)$) and let $a \in L(x)$. 
If there is no out-neighbor (in-neighbor)
of $a$ in $L(y)$ then remove $a$ from $L(x)$. Repeat this until a list becomes empty or no more changes can be made. Note that if we end up with an empty list after arc consistency then there is no homomorphism of $D$ to $H$.

Let $a_1,a_2,a_3,\dots,a_p$ be a min-max-ordering $<$ of the target digraph $H$. Define $\ell^+(i)$ to be the smallest subscript $j$ such that $a_j$ is an out-neighbor of $a_i$, if such $j$ exists. Furthermore,  define $\ell^-(i)$ to be the smallest subscript $j$ such that $a_j$ is an
in-neighbor of $a_i$, if such $j$ exists.

Consider the following linear program. For every vertex $v$ of $D$ and every vertex $a_i$ of $H$ define a variable $v_{i}$.  We also define 
a variable $v_{p+1}$ for every $v \in D$ whose value is set to zero.
\begin{tcolorbox}[colback=white]
\begin{center}
\begin{tabular}{r l l l}
    \text{min}  &   $\sum\limits_{v,i} c(v,a_i)(v_i-v_{i+1})$ & & \\
    \text{subject to:}  &   $v_i\geq  0$ & &\newtag{(C1)}{C1}\\
    & $v_1=  1$  &&  \newtag{(C2)}{C2} \\
    & $v_{p+1} =  0$ && \newtag{(C3)}{C3}\\
    & $v_{i+1} \leq  v_{i}$ &&\newtag{(C4)}{C4}\\
    & $v_{i+1} =  v_{i}$   &\text{if } $a_i \not\in L(v)$&\newtag{(C5)}{C5}\\
    & $u_{i} \leq  v_{l^+(i)}$  &$\forall uv \in A(D)$&\newtag{(C6)}{C6}\\
    & $v_{i} \leq  u_{l^-(i)}$   &$\forall uv \in A(D)$&\newtag{(C7)}{C7}
\end{tabular}
\end{center}
\end{tcolorbox}

Let us denote the set of constraints of the above LP by $\mathcal{S}$. In what follows, we prove that there is a one-to-one correspondence between integer solutions of $\mathcal{S}$ and homomorphisms from $D$ to $H$ when $H$ admits a min-max-ordering.
\begin{theorem}\label{LP-minmax}
If digraph $H$ admits a min-max-ordering,  then there is a one-to-one correspondence between homomorphisms of $D$ to $H$ and integer solutions of $\mathcal{S}$.
\end{theorem}
\begin{proof}
For homomorphism $f:D \to H$, if $f(v)=a_t$ we set $v_i=1$ for all $i \leq t$, otherwise we set $v_i=0$. We set $v_1=1$ and $v_{p+1}=0$ for all $v\in V(D)$. Now all the variables are nonnegative and we have $v_{i+1}\leq v_i$. Note that if $a_i\not\in L(v)$ then $f(v)\neq a_i$ and we have $v_i-v_{i+1}=0$. It remains to show that $u_i\leq v_{l^+(i)}$ for every $uv$ arc in $D$. Suppose for contradiction that $u_i=1$ and $v_{l^+(i)}=0$ and let $f(u)=a_r$ and $f(v)=a_s$. This implies that $u_r=1$, whence $i\leq r$; and $v_s=1$, whence $s<l^+(i)$. Since $a_ia_{l^+(i)}$ and $a_ra_s$ both are arcs of $H$ with $i \leq r$ and $s<l^+(i)$, the fact that $H$ has a min-ordering implies that $a_ia_s$ must also be an arc of $H$, contradicting the definition of $l^+(i)$. The proof for $ v_{i}  \leq  u_{l^-(i)}$ is analogous. 

Conversely, if there is an integer solution for $\mathcal{S}$, we define a homomorphism $f$ as follows: we let $f(v) = a_i$ when $i$ is the largest subscript with $v_i = 1$. We prove that
this is indeed a homomorphism by showing that every arc of $D$ is mapped to an arc of $H$. Let $uv$ be an arc of $D$ and assume $f(u)=a_r$, $f(v)=a_s$. We show that $a_ra_s$ is an arc in $H$. Observe that $1=u_r\leq v_{l^+(r)}\leq 1$ and $1=v_s\leq u_{l^-(s)}\leq 1$, therefore we must have $v_{l^+(r)}=u_{l^-(s)}=1$. Since $r$ and $s$ are the largest subscripts such that $u_r=v_s=1$ then $l^+(r)\leq s$ and $l^-(s)\leq r$. Since $a_ra_{l^+(r)}$ and $a_{l^-(s)}a_s$ are arcs of $H$, we must have the arc $a_ra_s$, as $H$ admits a max-ordering.

Furthermore, $f(v) = a_i$ if and only if $v_i = 1$ and $v_{i+1} = 0$, so, $c(v, a_i)$
contributes to the sum if and only if $f(v) = a_i$.
\end{proof}
We have translated the minimum cost homomorphism problem to a linear program. In fact, this linear program corresponds to a minimum cut problem in
an auxiliary network, and can be solved by network flow algorithms \cite{mincostungraph,approximation-2011}. In \cite{arash-esa}, a similar result to Theorem \ref{LP-minmax} was proved for the MinHOM(H) problem on undirected graphs when target the graph $H$ is bipartite and admits a min-max-ordering. We
shall enhance the above system $\mathcal{S}$ to obtain an approximation algorithm for the
case where $H$ is only assumed to have a min-ordering.

\section{LP for Digraphs with a min-ordering}\label{min-to-max}

In the rest of the paper assume lists are not empty. Moreover, non-empty lists guarantee a homomorphism when $H$ admits a min-ordering. For the sake of completeness we present the proof of the following lemma. The argument is simple and perhaps could have appeared in earlier literature.  
\begin{lemma}[\cite{hombook}]
\label{lemma-hom}
Let $H$ be a digraph that admits a min-ordering. If all the lists are non-empty after arc consistency, then there exists a homomorphism from $D$ to $H$.
\end{lemma}
\begin{proof}
Let $a_1,a_2,\dots,a_p$ be a min-ordering of $H$.  
For every vertex $x$ of $D$, define $f(x)=a_i$ where $a_i$ is the smallest element (according to the ordering) in $L(x)$.  
We show that $f$ is a homomorphism from $D$ to $H$. Let $xy$ be an arc of $D$. Suppose $f(x)=a_i$ and $f(y)=a_j$. Because of the arc-consistency, there exist $a_{j'}$ in $L(y)$ such that $a_ia_{j'} \in A(H)$ and there exists $a_{i'} \in L(x)$ 
such that $a_{i'}a_j \in A(H)$. Note that $j \le j'$ and $i \le i'$. Since $a_1,a_2,\dots,a_p$ is a min-ordering, then $a_ia_j \in A(H)$ and $f(x)f(y) \in A(H)$. 
\end{proof}
Suppose $a_1, a_2, \cdots, a_p$ is a min-ordering of $H$.
Let $E'$ denote the set of all the pairs $(a_i,a_j)$ such
that $a_ia_j$ is not an arc of $H$, but there is an arc $a_ia_{j'}$ of $H$ with $j' < j$ and an arc $a_{i'}a_j$
of $H$ with $i' < i$. Let $E=A(H)$ and define $H'$ to be the digraph with vertex set $V(H)$ and arc set
$E \cup E'$. Note that $E$ and $E'$ are disjoint sets.

\begin{observation}\label{observation1}
The ordering $a_1, a_2, \cdots, a_p$ is a min-max-ordering of $H'$.
\end{observation}

\begin{proof}
We show that for every pair of arcs $e=a_ia_{j'}$ and $e'=a_{i'}a_j$ in $E \cup E'$, with $i'<i$ and $j'<j$,
both $g=a_ia_j$ and $g'=a_{i'}a_{j'}$ are in $E \cup E'$. If both $e$ and $e'$ are in $E$, $g \in E \cup E'$ and $g' \in E$.

If only one of the arcs $e,e'$, say $e$, is in $E'$, there is a vertex $a_{j''}$ with $a_ia_{j''} \in E$ and $j'' < j'$, and a vertex $a_{i''}$ with $a_{i''}a_{j'} \in E$ and $i'' < i$. Now, $a_{i'}a_j$ and $a_ia_{j''}$ are both in $E$, so $g \in E \cup E'$. We may assume that $i'' \neq i'$, otherwise $g' = a_{i''} a_{j'} \in E$. If $i'' < i'$, then $g' \in E \cup E'$ because $a_{i'}a_{j''} \in E$; and if $i'' > i'$, then $g' \in E$ because $a_{i'}a_j \in E$.

If both edges $e, e'$ are in $E'$, then the earliest out-neighbor of $a_i$ and earliest in-neighbor of $a_j$ in $E$ imply that $g \in E \cup E'$, and the earliest out-neighbors of $a_{i'}$ and earliest in-neighbor of $a_{j'}$ in $E$ imply that $g' \in E \cup E'$.
\end{proof}

\begin{observation}\label{observation2}
Let $e=a_ia_j \in E'$. Then $a_i$ does not have any out-neighbor in $H$
after $a_j$, or $a_j$ does not have any in-neighbor in $H$ after $a_i$.
\end{observation}
Observation~\ref{observation2} easily follows from the fact that $H$ has a min-ordering. Since $H'$ has a min-max-ordering, we can
form system of linear inequalities ${\cal S}$,  for $H'$ as described in Section \ref{sec-minmax}.
Homomorphisms of $D$ to $H'$ are in a one-to-one correspondence with integer
solutions of $\mathcal {S}$, by Theorem \ref{LP-minmax}. However, we are interested in homomorphisms of $D$ to
$H$, not $H'$. Therefore, we shall add further inequalities to ${\cal S}$ to ensure
that we only admit homomorphisms from $D$ to $H$, i.e., avoid mapping arcs of
$D$ to the arcs in $E'$. These inequalities ensure the arc consistency and pair consistency constraints. Also in finding a list homomorphism from $D$ to $H$, they also allow us  (if necessary) to shift the image of vertex $u \in D$ to an element to a smaller element in its list, using the min-ordering properties. 

For every arc $e = a_ia_j \in E'$ and every arc $uv \in A(D)$, by Observation \ref{observation2}, if $a_j$ has an in-neighbor after $a_i$ then inequalities \ref{C8} and \ref{C11} are added, else if $a_i$ has an out-neighbor after $a_j$ then inequalities \ref{C9} and \ref{C10} are added. If neither of the previous cases happens then \ref{C9} and \ref{C11} are added. 

Additionally, for every pair $(x,y) \in V(D) \times V(D)$ consider a list $L(x,y)$ initially to be 
$L(x) \times L(y)$. 

Perform \emph{pair consistency} procedure as follows. 
Consider three vertices $x,y,z \in V(D)$. For $(a,b) \in L(x,y)$ if there is no $c \in L(z)$ such that $(a,c) \in L(x,z)$ and $(c,b) \in L(z,y)$ then remove $(a,b)$ from $L(x,y)$. Repeat this until a pair list becomes empty or no more changes can be made.

\begin{tcolorbox}[colback=white]
\begin{center}
\scalebox{0.77}{
\begin{tabular}{l l l l l}
    $v_j \leq u_s $ & $+$ & $\sum\limits_{\substack {t<i \\ a_ta_j \in E \\ a_t\in L(u)}}(u_t-u_{t+1})$ & \text{$a_s$ is the first in-neighbor of $a_j$ after $a_i$ in $L(u) $}&\newtag{(C8)}{C8}\\
    $v_j \leq v_{j+1}$ & $+$ & $\sum\limits_{\substack{ t<i \\ a_ta_j\in E \\ a_t\in L(u)}}(u_t-u_{t+1})$ & \text{if $a_j$ has no in-neighbor after $a_i$} &\newtag{(C9)}{C9}\\
    $u_i \leq v_s$ & $+$ & $\sum\limits_{\substack {t<j \\ a_ia_t\in E \\ a_t\in L(v)}}(v_t-v_{t+1})$ & \text{$a_s$ is the first out-neighbor of $a_i$ after $a_j$ in $L(v)$ } & \newtag{(C10)}{C10}\\
    $u_i \leq u_{i+1}$ & $+$ & $\sum\limits_{\substack{ t<j\\ a_ia_t\in E \\ a_t\in L(v)}}(v_t-v_{t+1})$ & \text{if $a_i$ has no out-neighbor after $a_j$} &\newtag{(C11)}{C11}
\end{tabular}}
\end{center}
\end{tcolorbox}
 Here, we assume that after pair consistency procedure no pair list is empty, as otherwise there is no homomorphism of $D$ to $H$. Therefore, by pair consistency, add the following constraints for every $u\neq v$ in $V(D)$ and $a_i\in L(u)$:
\begin{center}
\begin{tabular}{l l l l}
$u_i-u_{i+1}$ &$\le$ & $\sum\limits_{\substack{j: \\(a_i,a_j) \in L(u,v)}}(v_j-v_{j+1})$ &\newtag{(C12)}{C12}
\end{tabular}
\end{center}

Let the system of linear equation $\mathcal{S}$ together with constraints  ~\ref{C8},~\ref{C9},~\ref{C10}, \ref{C11}, and \ref{C12} be denoted by $\widehat{\mathcal{S}}$.

\begin{lemma}\label{one-to-one-ex}
If $H$ admits a min-ordering, then there is a one-to-one correspondence between homomorphisms of $D$ to
$H$ and the integer solutions of $\widehat{\mathcal{S}}$.
\end{lemma}

\begin{proof} In the proof of Theorem \ref{LP-minmax} we shown that from an integer solution for  ${\cal S}$, one can obtain a homomorphism from 
$D$ to $H'$. Let $f$ be such a homomorphism. We show that $f$ is a homomorphism from $D$ to $H$. Let $uv$ be an arc of $D$ and  let $f(u) = a_i, f(v) = a_j$. We have $u_i = 1$, $u_{i+1}=0$, $v_j = 1$, $v_{j+1}=0$, and for all $a_ta_j\in E$ with $t < i$ we have $u_t - u_{t+1} = 0$. We show that $a_ia_j \in E$. If it is not the case, then one constrain among \ref{C8} and \ref{C9} must hold. If $a_s$ is the first in-neighbor of $a_j$ after $a_i$, then we will also have $u_s = 0$, and so inequality \ref{C8} fails. Else, if $a_j$ has no in-neighbor after $a_i$, then inequality \ref{C9} fails.

Conversely, suppose $f$ is a homomorphism from $D$ to $H$ (i.e., $f$ maps each edge of $D$ to
an edge of $H$). We show that the additional constraints hold. 

We first show that inequality~\ref{C8} holds (the reasoning for other inequalities is similar). 
Let $uv \in A(D)$, and let $a_ia_j \in E'$. Now if $v_j=0$ or $u_s=1$ then~\ref{C8} holds. Hence, we may assume $u_s=0$ and $v_j=1$. Note that $a_s \in L(u)$ is the first in-neighbor of $a_j$ after $a_i$. Now, according to the way we define $u_i,u_s,v_j$ when 
we have homomorphism $f$ from $D$ to $H$, we have $f(u)=a_r < a_i< a_s$ and $a_j$ is before or equal to $f(v)$ in the min-ordering.
 
Notice that $u_{r}-u_{r+1}=1$. 
Now if $a_j < f(v) $ then $a_rf(v), a_sa_j \in A(H)$ and min-ordering implies that 
$a_ra_j \in A(H)$. On the other hand, if $f(v)=a_j$ then $a_ra_j \in A(H)$. Therefore, 
$a_ra_j \in A(H)$ and $a_r\in L(u)$. 
Now the sum of $(u_t-u_{t+1})$ over all $a_t \le a_i$ such
that $a_t$ is in $L(u)$ and is an in-neighbor of $a_j$ is one. Therefore, the 
inequality~\ref{C8} holds.  

Note that if there is a homomorphism from $D$ to $H$ then inequality \ref{C12} is a necessary condition for having such a homomorphism.
\end{proof}

\section{Approximation for Digraphs with a min-ordering}\label{sec-Approx-ALG}

In what follows, we describe our approximation algorithm for MinHOM$(H)$ where the fixed digraph $H$ has a min-ordering. We start off with an overview of our algorithm, Algorithm~\ref{alg:app-MinHom}. The proofs of the correctness and approximation bound are postponed to later subsections.

\begin{algorithm}[tb]
\caption{Approximation MinHOM($H$)}\label{alg:app-MinHom}
\begin{algorithmic}[1]
\Procedure{Approx--MinHOM}{$H$}
\State Construct $H'$ from $H$ (as in Section \ref{sec-minmax}) 
\State Let $u_i$s be the (fractional) values returned by $\widehat{\mathcal{S}}$. 
\State Sample $X$ uniformly from $[0,1]$
\State For all $u_i$s: if $X\leq u_i$ let $u'_{i}=1$, else let $u'_{i}=0$\label{rounding}
\State  Let $f(u)=a_i$ where $i$ is the largest subscript with $u'_i=1$ 

\Comment{$f$ is a homomorphism from $D$ to $H'$}
\State $F\gets$ all arcs in $E'$ to which some arcs of $D$ are mapped by $f$

\State Sample $Y$ uniformly from $[0,1]$

\While {$\exists$ arc $a_ia_j \in F$ with $i+j$ being maximum }
\While{  $\exists uv\in A(D)$ with $f(u)=a_i$ and $f(v)=a_j$ }\label{call-shift}
\If{$f(v)$ does not have an in-neighbor after $f(u)$}
\State \textsc{Shift}$(f,v)$
\ElsIf{$f(u)$ does not have an out-neighbor after $f(v)$}
\State \textsc{Shift}$(f,u)$
\EndIf


\EndWhile\label{euclidendwhile}
\State Remove $a_ia_j$ from $F$.
\EndWhile 

\State \textbf{return} $f$\Comment{$f$ is a homomorphism from $D$ to $H$}
\EndProcedure
\end{algorithmic}
\end{algorithm}

Let $D$ be the input digraph together with a costs function $c$, and let $H$ be a fixed target digraph $H$, let $a_1,\cdots, a_p$ be a min-ordering of the vertices of $H$. Algorithm~\ref{alg:app-MinHom}, first constructs digraph $H'$ from $H$ as explained in Section~\ref{min-to-max}. By Observation~\ref{observation1}, $a_1,\dots, a_p$ is a min-max-ordering for $H'$. By Lemma~\ref{one-to-one-ex}, the integral solutions of $\widehat{\mathcal{S}}$ are in one-to-one correspondence to homomorphisms from $D$ to $H$. At this point, our algorithm will minimize the cost function over extended ${\cal S}$ in polynomial-time using
a linear programming algorithm. This will generally result in a fractional solution (Even though the original system $\mathcal{S}$ is known to be totally unimodular~\cite{mincostungraph,approximation-2011} and hence has integral optima, we have added inequalities, thus losing this advantage). We will obtain
an integer solution by a randomized procedure called {\em rounding}. Choose, uniformly at random, a random variable
$X \in [0,1]$, and define the rounded values $u'_i = 1$ when $u_{i} \geq X$ ($u_i$ is the returned value by the $\widehat{\mathcal{S}}$) and $u'_{i} = 0$
otherwise. It is easy to check that the rounded values satisfy the original inequalities, i.e., correspond to a homomorphism $f$ of $D$ to $H'$.

\begin{algorithm}[ht]
\caption{The shifting procedure}\label{alg:shift}
\begin{algorithmic}[1]
\Procedure{Shift}{$f,x$}
\State Let Q be a Queue, $Q.enqueue(x)$
\While {Q is not empty}
\State $v\gets Q.dequeue( )$
\For{$uv \in A(D)$ with $f(u)f(v)\not\in A(H)$ \textbf{or} 

\hspace{3mm} $vu \in A(D)$ with $f(v)f(u)\not\in A(H)$}

\Comment{Here we assume the first condition hold, the other case is similar}

\Comment{Plus, we assume $f(v)$ does not have an in-neighbor after $f(u)$}

\State Let $t_1 < \dots < t_k$ be indices s.t. $a_{t_j}<f(v), a_{t_j}\in L(v),f(u)a_{t_j}\in A(H)$ \label{round_index}
\State Let $P_v\gets \sum\limits_{j=1}^{j=k} (v_{t_j} - v_{t_j+1})$ and $P_{v,t}\gets ({v_t-v_{t+1}})\mathbin{/}{P_v}$
\If{$\sum\limits_{\small p=1}^{\small q} P_{v,{t_p}} < Y \leq \sum\limits_{\small p=1}^{\small q+1} P_{v,{t_p}}$}  
\State $f(v)\gets a_{t_q}$, set $v'_{i}=1$ for $1\leq i\leq t_q$, and set $v'_{i}=0$ for $t_p< i$
\EndIf
\For{ $vz \in A(D)$ ($zv \in A(D)$) with $f(v)f(z)\not\in A(H)$ 

\hspace{7mm} ($f(z)f(v) \not\in A(H)$) }
\State $Q.enqueue(z)$
\EndFor
\EndFor
\EndWhile
\State \textbf{return} $f$\Comment{$f$ is a homomorphism from $D$ to $H'$}
\EndProcedure

\end{algorithmic}
\end{algorithm}

Now the algorithm will modify the solution $f$ to become a homomorphism from $D$ to
$H$, i.e., to avoid mapping the arcs of $D$ to the arcs in $E'$. This will be accomplished by another randomized procedure, which we call \textsc{Shift}, Algorithm~\ref{alg:shift}. We choose, uniformly at random, another random variable $Y \in [0,1]$, which will guide the shifting. Let $F$ denote the set of all arcs in $E'$ to which some arcs of $D$ are mapped by $f$. If $F$ is empty, we need no shifting. Otherwise, let $a_ia_j$ be an arc of $F$ \textbf{where $i+j$ is maximum}.  Since $F \subseteq E'$, Observation \ref{observation2}
implies that either $a_j$ has no in-neighbor after $a_i$ or $a_i$ has no out-neighbor after $a_j$. Suppose the first case happens (the shifting process is similar in the other case).

Consider a vertex $v$ in $D$ such that $f(v)=a_j$ (i.e. $v_j'=1$ and $v_{j+1}'=0$) and $v$ has an in-neighbor $u$ in $D$ with $f(u)=a_i$ (i.e. $u_i'=1$ and $u_{i+1}'=0$). 
For such a vertex $v$, let $S_v=\{a_{t_1},a_{t_2},\dots,a_{t_k}\}$ be the set of all vertices $a_t$ with $t < j$ such that $a_ia_t \in E$ and $a_t\in L(v)$. We will show in Lemma~\ref{set-index-shifting} that $S_v$ is not empty. 
Suppose $S_v$ consists of $a_t$ with subscripts $t$ ordered as $t_1 < t_2 < \dots < t_k$. 
The algorithm now selects one vertex from this set as follows. Let $P_{v,t} = \frac{v_{t} - v_{t+1}}{P_v},$ where 
\begin{align*}
P_v = \sum\limits_{\substack{t : a_t\in S_v}} (v_{t} - v_{t+1}).
\end{align*}
Note that $P_v > 0$ because one among the constraints \ref{C10} and \ref{C11} holds. Then $a_{t_q}$ is selected if $\sum\limits_{\small p=1}^{\small q} P_{v,{t_p}} < Y \leq \sum\limits_{\small p=1}^{\small q+1} P_{v,{t_p}}$. Thus a concrete $a_t$ is selected with probability $P_{v,t}$, which is proportional to the difference of the fractional values $v_t-v_{t+1}$. When the selected vertex is $a_t$, we shift the image of the vertex $v$ from $a_j$ to $a_t$, and set $v'_r=1$ if $r\leq t$, else set $v'_r=0$. Note that $a_t$ is before $a_j$ in the min-ordering\footnote{The images are always shifted towards smaller elements.}. 
Now we might need to shift images of the neighbors of $v$. In this case, repeat the shifting procedure for neighbors of $v$. Let $z$ be an out-neighbor of $v$ and the case with $z$ being an in-neighbor of $v$ is handled similarly. First suppose $a_jf(z)$ is an arc
of $H$. If $a_t$ has an out-neighbor after $f(z)$
then by min-ordering property $a_tf(z)$ is an arc of $H$ and hence, the image of 
$z$ does not need to be changed. If $a_t$ has no out-neighbor after $f(z)$ in $L(z)$, then all the out-neighbors of $a_t$ are before $f(z)$ and an out-neighbor of $a_t$ in $L(z)$ is selected according to random variable $Y$ as explained above. Moreover, constraints \ref{C12} ensures that $P_z> 0$.
If $a_jf(z)$ is not an arc of $H$, then $a_jf(z)$ is an arc of $H'$. 
Now, we change the image of $z$ from $f(z)$ by selecting a vertex in $L(z)$ that is an 
out-neighbor of $a_t$ in $L(z)$ according to random variable $Y$. Again in this case, constraints \ref{C10} or \ref{C11} ensure that $P_z >0$. 

This processes continues in a Breadth-first search (BFS) like manner, until no more shift is required. See Figure \ref{shift-pic} for an illustration and Example~\ref{Example:alg-min} for detailed explanation. Note that a vertex might be visited multiple times in procedure shift while a pair $(v, a_i) \in V (D) \times V (H)$ is considered at most one time.

\begin{example}[Examples for Algorithm~\ref{alg:app-MinHom}]
\label{Example:alg-min}
    Here, we provide detailed explanation on the examples given in Figure~\ref{shift-pic}. In the right example, the target digraph is $H_1$ and the input is $D_1$. The right digraphs ($D_1,H_1$) both can be viewed as bipartite graphs and $1,2,3,4,5,6,7$ is a min-ordering of $H$. When $x$ is mapped to $3$ and $w$ is mapped to $6$ then the algorithm should shift the image of $w$ from $6$ to $5$ and since $35$ is an arc there is no need to shift the image of $y$. 
In the left example, the target digraph is $H$ and the input is $D$. In $H$, $1,2,3,4,5,6,7,8$ is a min-ordering and $24$ is a missing arc.  Suppose $x$ is mapped to $2$, $y$ to $4$, $w$ to $7$, $z$ to $8$, $u$ to $5$ and $v$ to $2$. 
Then we should shift the image of $y$ to $3$ and then $w$ to $6$ and $z$ to $6$ and then $u$ to $3$ and $v$ to one of the $1,2$. 
\end{example}

\begin{figure}[tb]
\begin{center}
\includegraphics[height=5cm]{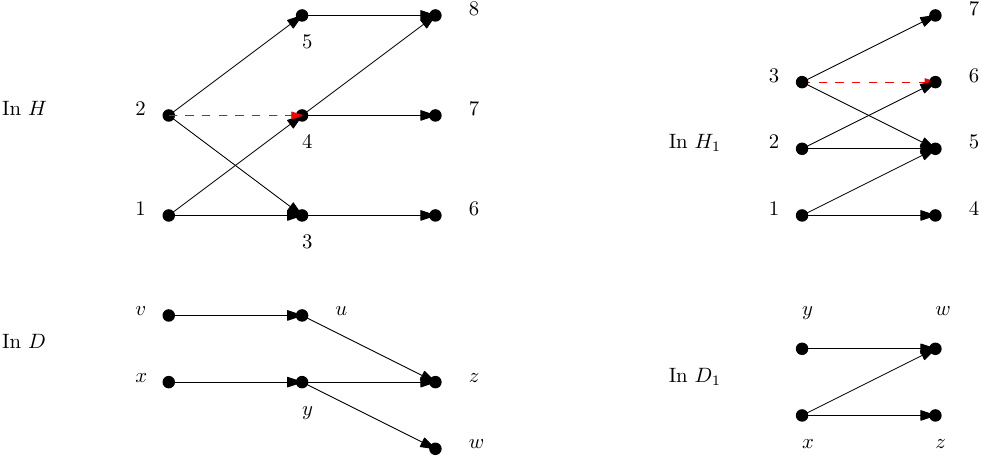}
\caption{Two examples for Algorithm~\ref{alg:app-MinHom}. See Example~\ref{Example:alg-min} for details on these examples. 
} \label{shift-pic}
\end{center}
\end{figure}

\begin{lemma}
\label{set-index-shifting}
During procedure \textsc{Shift}, the set of indices $t_1 < \dots < t_k$ considered in Line \ref{round_index} of the Algorithm \ref{alg:app-MinHom} is non-empty.
\end{lemma}
\begin{proof}
In procedure \textsc{Shift}, consider $vz$ such that $f(v)f(z)\not\in E(H')$ and $f(v)=a_t$ and $f(z)=a_l$. We have $0 < v_t-v_{t+1}$, and together with constraint \ref{C12}, this implies \begin{align*}
0 < v_t-v_{t+1} \leq \sum\limits_{\substack{j: \\(a_t,a_j) \in L(v,z)}}(z_j-z_{j+1}).    
\end{align*}
Therefore, there must be an index $l'$ such that $(a_t,a_{l'})\in L(v,z)$. It remains to show that there exists such an 
$a_{l'}$ appearing before $a_l$ in the min-ordering. There are two cases to consider. First is $f(v)$ is set to $a_t$ in rounding step (Line \ref{rounding}). Second is image of $v$ was shifted from $a_j$ to $a_t$ in procedure \textsc{Shift}.

For the first case, note that, since $f$ is a homomorphism from $D$ to $H'$,  $a_ta_l\in E(H')\setminus E(H)$. Arc $vz$ is mapped to $a_ta_l$ in rounding step (Line \ref{rounding}) according to the random variable $X$. Note that, during procedure \textsc{Shift}, we do not map any arc of $D$ to edges in $E(H')\setminus E(H)$. Therefore, we have $X\leq v_t,z_l$. Consider the situation where $a_l$ has no in-neighbor after $a_t$. Let $a_s$ be the first out-neighbor of $a_t$ after $a_l$, then we have $ z_s < X\leq v_t$. This together with inequality \ref{C10} implies that 
\begin{align*}
0 < \sum\limits_{\substack {l'<l \\ a_ta_l\in E \\ a_{l'}\in L(z)}}(z_{l'}-z_{l'+1}). 
\end{align*}
Hence, there exists an index $l'<l$ as we wanted. The argument for the case where $a_t$ has no out-neighbor after $a_l$ is similar. 

For the second case, before mapping $v$ to $a_t$, there was an index $a_j$ such that $a_t < a_j$. There are two cases regarding $a_ja_l$. Either it is in $E(H)$ or it is in $E(H')\setminus E(H)$. In both cases, $a_{l'}$ must appear before $a_l$ as otherwise, min-max-ordering implies $a_ta_l\in E(H')$, contradicting our assumption.
\end{proof}

\begin{theorem}
Algorithm~\ref{alg:app-MinHom}, runs in polynomial time, returns a homomorphism of $D$ to $H$.
\end{theorem}
\begin{proof}
It it easy to see that, if there exists a homomorphism from $D$ to $H$, then there is a homomorphism from $D$ to $H$ that maps every vertex of $D$ to the smallest vertex in its list (Lemma~\ref{lemma-hom}). We show that a sequence of shifting, either stops at some point, or it keeps shifting to a smaller vertex in each list. Lemma \ref{set-index-shifting} allows us to apply the shifting as long as it is necessary. Therefore,  on the latter case, after finite (polynomially many) steps, we end up mapping every vertex of $D$ to the smallest vertex in its list.

Consider an arc $vz\in A(D)$. Suppose $f(v)=a_t$ and $f(z)=a_l$ and $a_ta_l$ is an arc of $H$. Assume that we have shifted the image of $v$ from $a_t$ to $a_{t'}\in L(v)$ where $a_{t'}$ is before $a_t$ in the min-ordering. If $a_{t'}a_l$ is in $A(H)$ then we do not have to shift the image of $z$. Note that, since $a_{t'}$ is in $L(v)$ then it has to have an out-neighbor in $L(z)$. Let say $a_{l'}\in L(z)$ is an out-neighbor of $a_{t'}$. If $a_{l'}$ is after $a_l$ in the min-ordering then it implies $a_{t'}a_{l}\in A(H)$. Else, $a_{l'}$ is before $a_{l}$ in the min-ordering and we shift the image of $z$ to a smaller vertex in its list. 

Thus, the above argument shows that the shifting modifies the homomorphism $f$, and hence, the corresponding values of the variables. Namely, $v'_{t+1}, \dots, v'_{j}$ are reset to $0$, keeping all other values the same. It is important to note that these modified values still satisfy the original set of constraints $\mathcal{S}$, i.e., the modified mapping is still a homomorphism from $D$ to $H'$. Moreover, during the shifting, the number of arcs in $D$ that are mapped to arcs in $E'$ does not increase.

We repeat the same process for the next $v$ with these properties, until no arc of $D$ is mapped to an arc in $E'$. Each iteration involves at most $|V(H)|\cdot|V(D)|$ shifts. After at most $|E'|$ iterations, no edge of $D$ is mapped to an arc in $F$ and we no longer need to shift (See Figure~\ref{shift-pic} for an example). 
\end{proof}

\subsection{Analyzing the Approximation Ratio}
We now claim that, the cost of this homomorphism is at most $|V(H)|^2$ times the minimum cost of a homomorphism. Let $W$ denote the value of the objective function with the fractional optimum $u_i, v_j$, and $W'$ denote the value of the
objective function with the final values $u'_{i}, v'_{j}$, after the rounding and all the shifting.
Also, let $W^*$ be the minimum cost of a homomorphism of $D$ to $H$. Obviously, $W \leq W^* \leq W'$.

We now show that the expected value of $W'$ is at most a constant times $w$. Let us focus on
the contribution of one summand, say $v'_{t} - v'_{t+1}$, to the calculation of the cost.

In any integer solution, $v'_{t} - v'_{t+1}$ is either $0$ or $1$. The probability that
$v'_{t} - v'_{t+1}$ contributes to $W'$ is the probability of the event that $v'_{t} = 1$
and $v'_{t+1} = 0$. This can happen in the following situations:
\begin{enumerate}
\item 
$v$ is mapped to $a_t$ by rounding, and is not shifted away. In other words, we have
$v'_{t} = 1$ and $v'_{t+1} = 0$ after rounding, and these values do not change by
procedure \textsc{Shift}.
\item
$v$ is first mapped to some $a_j, j > t$, by rounding, and then re-mapped to $a_t$ by procedure \textsc{Shift}.


\end{enumerate}

\begin{lemma}\label{pr-bound}

The expected contribution of one summand, say $v'_t-v'_{t+1}$, to the expected cost of $W'$ is at most $|V(H)|^2c(v,a_t)(v_t-v_{t+1})$.
\end{lemma}

\begin{proof}
Vertex $v$ is mapped to $a_t$ in two cases. The first case is where $v$ is mapped to $a_t$ by rounding (Line \ref{rounding}) and is not shifted away. In other words, we have
$v'_{t} = 1$ and $v'_{t+1} = 0$ after rounding, and these values do not change by procedure \textsc{Shift}. Hence, for this case we have:
\begin{align*}
\P[f(v)=a_t] &=\P[v_{t+1} < X \leq v_t]\cdot \P[v\text{ is not shifted in procedure \textsc{Shift}}]\\
&\leq v_t-v_{t+1}
\end{align*}
Whence this situation occurs with probability at most $v_t-v_{t+1}$, and the expected contribution is at most $c(v,a_t)(v_t-v_{t+1})$.

Second case is where $f(v)$ is set to $a_t$ during procedure \textsc{Shift}. The algorithm calls \textsc{Shift} if there exists $u^0u^1\in A(D)$ such that $f(u^0)f(u^1)\in E(H')\setminus E(H)$ (Line \ref{call-shift}). Let us assume it calls \textsc{Shift}$(f,u^1)$. Procedure \textsc{Shift} modifies images of vertices $u^1,u^2,\dots$. Consider the last time that \textsc{Shift} changes image of $v$. Note that $u^1,\cdots ,u^k=v$ is an oriented walk, meaning that there is an arc between every two consecutive vertices of the sequence and the $u^i$s are not necessarily distinct. 

\paragraph{Base case} We first compute the contribution for a fixed $j$, that is the contribution of shifting $v$ from a fixed $a_j$ to $a_t$. We use induction on $k$. 
Consider the simplest case where $k=1$. In this case $v$ is first mapped to $a_j, j > t$, by rounding, and then re-mapped to $a_t$ during procedure \textsc{Shift}.
This happens if there exist $i$ and $u$ such that $uv$ is an arc of $D$ mapped to $a_ia_j \in F\subseteq E'$ with $i+j$ being the maximum,
and then the image of $v$ is shifted to $a_t$ ($a_t<a_j$ in the min-ordering), where $a_ia_t \in E=A(H)$. In other words, we have
$u'_{i} = v'_{j} = 1$ and $u'_{i+1} = v'_{j+1} = 0$ after rounding (Line \ref{rounding}); and then $v$ is shifted from $a_j$ to $a_t$. According to the algorithm and the way we process the arcs in $F\subseteq E'$,  
we may assume that $v$ has not been 
shifted because of some $i'>i$ and some in-neighbor $u'$ of $v$ at the moment. Moreover, note that many in-neighbors $u$ of $v$ could cause such a shift. 
However, we have the corresponding inequality for each of them. That is we must have $w_{i+1}<X$ for every in-neighbor $w$ of $v$. To see this, consider the case where $a_{i+1}$ has some out-neighbor before $a_j$, then $a_{i+1}a_j$ is also an arc in $E'$. In this case we must have $w_{i+1} < X$ for every in-neighbor $w$ of $v$, otherwise, we have $wv$ being mapped to $a_{i+1}a_j$ by the rounding which is a contradiction to the choice of $a_ia_j$. Now suppose $a_{i+1}$ has no out-neighbor before $a_j$. In this case according to constraint~\ref{C6} we have $u_{i+1} \le v_{j+1}$ (and in particular for every other in-neighbor $w$ of $v$, we have $w_{i+1} \le v_{j+1}$), and hence, $w_{i+1} <X$. This argument implies that we only need to consider one in-neighbor for $v$, namely $u$ to compute the probability of one arc of $D$ being mapped to $a_ia_j$.

First assume that $a_i$ has some out-neighbor after $a_j$. Then we have

\begin{allowdisplaybreaks}
\begin{align*}
\P[u'_i=v'_j=1 , u'_{i+1}=v'_{j+1}=0]&=\P[\max \{u_{i+1},v_{j+1}\} < X\leq \min \{u_i,v_j \}]\\ 
&= \min \{u_i,v_j \} - \max \{u_{i+1},v_{j+1}\}\\
&\leq u_i - v_{j+1} \leq u_i-v_s \\
&\leq \sum\limits_{\substack{t<j \\ a_ia_t \in E \\ a_t\in L(v)}} (v_{t} - v_{t+1}) = P_v,
\end{align*}
\end{allowdisplaybreaks}
where in the first inequality $a_s\in L(v)$ is the first out-neighbor of $a_i$ after $a_j$ and $v_s \leq v_{j+1}$ due to ~\ref{C4}, and the last inequality follows from inequality~\ref{C10}. 

Second, for the case where $a_i$ has no out-neighbor after $a_j$, we have 

\begin{align*}
\P[u'_i=v'_j=1 , u'_{i+1}=v'_{j+1}=0]&=\P[\max \{u_{i+1},v_{j+1}\} < X\leq \min \{u_i,v_j \}]\\
&= \min \{u_i,v_j \} - \max \{u_{i+1},v_{j+1}\}\\
&\leq u_i - u_{i+1} \\
&\leq \sum\limits_{\substack{t<j \\ a_ia_t \in E \\ a_t\in L(v)}} (v_{t} - v_{t+1}) = P_v,
\end{align*}
where the last inequality follows from inequality~\ref{C11}.

Having $uv$ mapped to $a_ia_j$ in the rounding step, we shift $v$ to $a_t$ with probability $P_{v,t}={\frac{(v_t - v_{t+1})}{P_v}}$. Note that the upper bound $P_v$ is independent from the choice of $u$ and $a_i$. Therefore, for a fixed $a_j$, the probability that $v$ is shifted from $a_j$ to $a_t$ is at most $ \frac{v_t - v_{t+1}}{P_v}\cdot P_v = v_t - v_{t+1}$.

\paragraph{Inductive step} For $k>1$, consider the oriented walk $u^0,\cdots ,u^k=v$. Before calling \textsc{Shift}$(f,u^1)$, this walk is mapped to some vertices in $H$. Without loss of generality, let us assume these vertices are $a_{\ell_0},a_{\ell_1},\cdots, a_{\ell_k}$. Note that $a_{\ell_i}$s may not be distinct. Once again we compute the contribution for a fixed $\ell_k=j$, that is the contribution of shifting $v$ from a fixed $a_{\ell_k}=a_j$ to $a_t$.

The algorithm calls \textsc{Shift}$(f,u^1)$ and, in procedure \textsc{Shift}, images of $u^1,u^2,$ $\dots,u^k=v$ are changed in this order. We are interested in probability of mapping $v$ from fixed $a_{\ell_k}=a_j$ to $a_t$. Analyzing the situation for $u^1$ is the same as the case for $k=1$. As induction hypothesis, assume for $u^1,\cdots, u^{k-1}$, the probability that the algorithm shifts image of $u^{k-1}$ to some $a_{\ell_{k-1}}$ is at most $u^{k-1}_{\ell_{k-1}}-u^{k-1}_{\ell_{k-1}+1}$. At this point $f(u^{k-1})=a_{\ell_{k-1}}$ and $f(v)=a_{\ell_k}$. Without loss of generality, assume $u^{k-1}v\in A(D)$ and the case where $vu^{k-1}\in A(D)$ is handled similarly. Note that $a_{\ell_{k-1}}a_{\ell_k}$ is not an arc in $A(H)$ as otherwise no change is required for image of $v$. Define set of indices
\begin{align*}
    B=\{r \mid r<\ell_k, a_{\ell_{k-1}}a_r \in A(H) , a_r \in L(v)\}.
\end{align*} 
The algorithm chooses $a_t$ where $a_t\in L(v),a_t<a_{\ell_k}$ and $a_{\ell_{k-1}} a_t\in A(H)$ with probability
 \begin{align*}
 \frac{v_t-v_{t+1}}{\sum\limits_{\substack{r\in B}} (v_{r} - v_{r+1})}
 \end{align*}

Let assume previous image of $u^{k-1}$ was $a_{\ell'_{k-1}}$ and it shifted to $a_{\ell_{k-1}}$. There are two cases to consider, namely $a_{\ell'_{k-1}}a_{\ell_k}$ is an arc of $H$ and it is an arc in $E'$. 

First suppose $a_{\ell'_{k-1}}a_{\ell_k}$ is an arc of $H$ when $u^{k-1}u^k \in A(D)$. Now we observe that $a_{\ell_{k-1}}$ does not have any out-neighbor $a_s$ with $s \ge \ell_k$. This is because $a_{\ell_{k-1}}a_{s},$ $a_{\ell_{k-1}}a_s \in A(H)$ and the min-ordering implies $a_{\ell_{k-1}}a_{\ell_k}\in A(H)$ which contradicts our assumption. Thus, by inequality \ref{C11} or \ref{C12} we get 
\begin{align*}
    u_{\ell_{k-1}}-u_{\ell_{k-1}+1}\leq \sum\limits_{r\in B} (v_{r} - v_{r+1}).
\end{align*} 
Therefore, by the induction hypothesis and the above inequality, the probability that our algorithm shifts the image of $v$ from $a_{\ell_k}$ to $a_t$ is at most
\begin{align*}
    \frac{(u_{\ell_{k-1}}-u_{\ell_{k-1}+1})(v_t-v_{t+1})}{\sum\limits_{r\in B} (v_{r} - v_{r+1})}\leq  v_{t} - v_{t+1}
\end{align*}

In the second case suppose $a_{\ell'_{k-1}}a_{\ell_k}$ is not an arc of $A(H)$, and hence, is an arc of $H'$. Notice that $a_{\ell'_{k-1}}$ has an out-neighbor (in-neighbor) $a_s$ before $a_{\ell_k}$ and $a_{\ell_k}$ has an in-neighbor $a_{s'}$ before $a_{\ell'_{k-1}}$. If $a_{\ell_{k-1}}$ has no out-neighbor (in-neighbor) after $a_{\ell_k}$, thus by inequality \ref{C11}, we get 
\begin{align*}
    u_{\ell_{k-1}}-u_{\ell_{k-1}+1}\leq \sum\limits_{r\in B} (v_{r} - v_{r+1})
\end{align*} 
Therefore, by the induction hypothesis and the above inequality, the probability that our algorithm shifts the image of $v$ from $a_{\ell_k}$ to $a_t$ is at most
\begin{align*}
    \frac{(u_{\ell_{k-1}}-u_{\ell_{k-1}+1})(v_t-v_{t+1})}{\sum\limits_{r\in B} (v_{r} - v_{r+1})}\leq  v_{t} - v_{t+1}
\end{align*}

Hence, we may assume that $a_{\ell_{k-1}}$ has an out-neighbor $a_q$ after $a_{\ell_k}$, and therefore $a_{s'}$ must be before $a_{\ell_{k-1}}$ as otherwise the min-ordering property would imply $a_{\ell_{k-1}}a_{\ell_k}$ is an arc in $H$. Hence, in this case, $a_{\ell_{k-1}}a_{\ell_k}$ is an arc in $E'$. Notice that we had $v_q,v_{\ell_k+1},u_{\ell'_k+1} <X \le v_{\ell_k},u_{\ell'_{k-1}},u_{\ell_{k-1}}$. Thus, the probability that this situation occurs is at most 
$u_{\ell'_{k-1}}-v_{\ell_k+1} \le u_{\ell_{k-1}}-v_{q}$ (due to inequality \ref{C4}). Now by inequality \ref{C10}, we get 
\begin{align*}
    u_{\ell_{k-1}}-v_{q} \leq \sum\limits_{r\in B} (v_{r} - v_{r+1})
\end{align*} 
Therefore, by the above inequality, the probability that our algorithm shifts the image of $v$ from $a_{\ell_k}$ to $a_t$ is at most
\begin{align*}
    \frac{(u_{\ell_{k-1}}-v_{q})(v_t-v_{t+1})}{\sum\limits_{r\in B} (v_{r} - v_{r+1})}\leq  v_{t} - v_{t+1}
\end{align*}
This completes this part of the proof.

Let $L(v)=\{a^v_1\cdots,a^v_k\}$. Clearly, during procedure \textsc{Shift}, image of $v$ can be shifted to $a^v_i$ from any of vertices $a^v_{i+1},\cdots, a^v_k$. For any fixed $a_j\in\{a^v_{i+1},\cdots, a^v_k\}$, this shift is initiated from vertices in $V(H)$ that are incident with some edges in $E'$, and reaches to $a_j$ to shift image of $v$. Shifting of image of $v$ happens because of missing edges from $a_j$ that is at most $|V(H)|-d^+(a_j)-d^-(a_j)\leq |V(H)|$ ($d^+(a_j)$ and $d^-(a_j)$ are out-degree and in-degree of $a_j$ respectively). Therefore, the contribution of $v$ and $a^v_i$ to the expected value of $W'$ is at most $(1+|V(H)|(k-i))c(v,a^v_i)(v_{a^v_i}-v_{a^v_{i+1}})$ where $(v_{a^v_i}-v_{a^v_{i+1}})$ is the upper bound on the probability provided before. Recall that $W'$ is the value of the
objective function with the final values $u'_{i}, v'_{j}$, after the rounding and all the shifting.
\end{proof}

\begin{theorem}\label{de-random}
Algorithm \ref{alg:app-MinHom} returns a homomorphism with the expected cost $|V(H)|^2$ times the optimal cost. The algorithm can be de-randomized to obtain a deterministic $|V(H)|^2$- approximation algorithm.
\end{theorem}
\begin{proof} 
By Lemma~\ref{pr-bound} the expected value of $W'$ is
\begin{align*}
\mathbb{E}[W']&=\mathbb{E}\left[\sum\limits_{v,i} c(v,a_i)(v'_i-v'_{i+1})\right]\\
&=\sum\limits_{v,i} c(v,a_i)\mathbb{E}[v'_i-v'_{i+1}]\\
&\leq|V(H)|^2\sum\limits_{v,i} c(v,a_i)(v_i-v_{i+1})\\
&\leq |V(H)|^2W\\
&\leq |V(H)|^2W^*. 
\end{align*}
At this point we have proved that Algorithm~\ref{alg:app-MinHom} produces a homomorphism whose
expected cost is at most $|V(H)|^2$ times the minimum cost. It can be transformed to a deterministic algorithm as follows. There are only polynomially many values $v_t$ (at most $|V (D)|\cdot|V (H)|$). When $X$ lies anywhere between two such consecutive values, all computations will remain the same. Thus we can de-randomize the first phase by trying all these values
of $X$ and choosing the best solution. Similarly, there are only polynomially many values of the partial sums $\sum\limits_{i=1}^{p} P_{u,t_i}$ (again at most $|V(D)|\cdot|V (H)|$), and when $Y$ lies between two such consecutive values, all computations remain the same. Thus we can also de-randomize the second phase by trying all possible values and choosing the best. Since the expected value is at most $|V(H)|^2$ times the minimum cost, this bound also applies to this best solution.
\end{proof}
\section{Approximation for Digraphs with a k-min-ordering}\label{extended-min-section}
Digraphs admitting $k$-min-ordering ($k > 1$) do not admit a min-ordering or a conservative majority polymorphism. However, this does not rule out the possibility of a constant factor approximation algorithm. We show that they are in fact DAT-free digraphs and the \textsc{List Homomorphism} problem is polynomial-time solvable for this class of digraphs.

It turns out that digraphs admitting a $k$-min-ordering do not contain a DAT. Furthermore, \textsc{List Homomorphism} problem is polynomial-time solvable for this class of digraphs (Lemmas~\ref{List-extended-min-ordering-1} and \ref{List-extended-min-ordering-2}).  

In the rest of this section $\overrightarrow{C_k}$ denotes an induced directed cycle with vertices $\{0,1,\dots,k-1\}$ and the arc set  $\{01,12,\dots,(k-2)(k-1),(k-1)0 \}$.

\begin{lemma}\label{List-extended-min-ordering-2}
Let $H$ be a digraph that admits a $k$-min-ordering. Then LHOM($H$) is polynomial-time solvable.
\end{lemma}
\begin{proof}
 Let $D,H,L$ be an instance of LHOM($H$) where $D$ is the input digraph and $L$ is the set of lists, i.e. for every $x \in V(D)$, $L(x) \subseteq V(H)$. We run the arc consistency procedure and suppose the lists are not empty after the arc consistency procedure. Let $V_0,V_1,\dots,V_{k-1} $ be the sets of vertices of $H$ according to the $k$-min-ordering $<$. We also note that if there exists a homomorphism $\phi: V(D) \rightarrow V(H)$, then $D$ must be homomorphic to $\overrightarrow{C_k}$ because $H$ is homomorphic to $\overrightarrow{C_k}$. This means the vertices of $D$ are partitioned into $D_0,D_1,\dots,D_{k-1}$ where the arcs of $D$ go from some $D_i$ to $D_{i+1}$, $0 \le i \le k-1$ (sum modulo $k$). For simplicity we may assume that $D$ is weakly connected; i.e. the underlying graph of $D$ is connected. Moreover, without loss of generality let $x$ be an arbitrary vertex in $D_0$ ($D_0$ is not empty). Now the vertices of $D_0$ are mapped to some $V_{\ell}$, for some $0 \le \ell \le k-1$. In other words, $L(x) \cap V_{\ell} \ne\emptyset$.




Now for every $y \in D_{j}$ and every $0 \le j \le k$, set $f(y)$ to be the smallest element in $L(y) \cap V_{j+\ell}$ according to $<$. Observe that the restriction of $<$ on $V_i \cup V_{i+1}$, $0 \le i \le k-1$, is a min-ordering. Suppose $yz$ is an arc of $D$ with $y \in D_{j}$ and $z \in D_{j+1}$. We show that $f(y)f(z)$ is an arc of $H$. Suppose $f(y)=a$ and $f(z)=b$. Since we run the arc-consistency procedure, there exists some element $b' \in L(z) \cap V_{\ell+j+1}$ such that $ab' \in A(H)$, and there exists some $a' \in L(y) \cap V_{\ell+j}$ so that $a'b \in A(H)$. The ordering $<$ on $V_{\ell+j} \cup V_{\ell+j+1}$ is a min-ordering, and hence, $ab$ is an arc of $H$.
\end{proof}
\begin{lemma}\label{List-extended-min-ordering-1}
Let $H$ be a digraph that admits a $k$-min-ordering. Then $H$ does not contain a DAT.
\end{lemma}
\begin{proof} 
It was shown in \cite{hell2011dichotomy} that digraph $H_1$ is DAT-free if and only if $V(H_1) \times V(H_1)$ can be partitioned into two sets $V_f,V_g$ where there exist two polymorphisms $f,g$ over $H_1$ such that 
$f$ is a semilattice on $V_f$ and $g$ is a majority over $V_g$. Let $V_0,V_1,\dots,V_{k-1} $ be the vertices of $H$ according to the $k$-min-ordering $<$. Define the binary polymorphism $f$ over $H$ as follows. 
\begin{enumerate}
    \item $f(x,y)=f(y,x) =x$ when $x,y \in V_i$ and $x<y$ (in the ordering $<$),
    \item $f(x,y)=x$, $f(y,x)=y$ when $x \in V_i$ and $y \in V_j$, $0 \le i \ne j \le k-1$,
    \item $f(x,x)=x$ for every $x \in V(H)$.
\end{enumerate}

First we show that $f$ is a polymorphism on $H$ and it is semilattice on $V_f = \{(x,y) \mid x,y \in V_i, 0 \le i \le k-1\}$. Let $xx',yy' \in A(H)$ where $x,y \in V_i$ and $x',y' \in V_{i+1}$. Now $f(x,y)f(x',y') \in A(H)$ because between $V_i,V_{i+1}$ we have a min-ordering, implying that $f$ is a polymorphism. It is also easy (since $<$ is min-ordering) to see that $f$ is associative. Now, define the ternary polymorphism $g$ over $H$ as follows :  
\begin{enumerate}
    \item $g(x,y,z)=x$ when $x,y,z \in V_i$,
    \item $g(x,y,z)=x$ when $x \in V_i$, $y \in V_j$, $z \in V_{\ell}$ and $i,j,\ell$ are all distinct,
    \item $g(x,y,z)=g(z,x,y)=g(x,z,y)=g(z,y,x)=g(y,z,x)=g(y,x,z)=x$ when $x,y \in V_i$, $x < y$ (in the ordering $<$), and $z \in V_j$, $i \ne j$,
    \item $g(x,x,y)=g(x,y,x)=g(y,x,x)=x$ when $x \in V_i$ and $y \in V_j$, $i \ne j$, 
    \item $g(x,x,x)=x$ for all $x \in V(H)$. 
\end{enumerate}

We show that $g$ is a polymorphism over $H$, and therefore, it is a majority polymorphism over the pairs in $V_g=\{ (x,y) \mid x \in V_i, y \in V_j, i \ne j \}$. By definition, we need to show that
\begin{align*}
    \forall xx',yy',zz'\in A(H)\implies g(x,y,z)g(x',y',z') \in A(H)
\end{align*}

\CCase{1}  If $x,y,z$ all belong to the same $V_i$, then $x',y',z' \in V_{i+1}$, and hence, by definition
\begin{align*}
    g(x,y,z)g(x',y',z') = xx' \in A(H).
\end{align*}

\CCase{2} If $x,y,z$ belong to three different partite sets, then $x',y',z'$ belong to three distinct partite sets, and hence, $$g(x,y,z)g(x',y',z')=xx' \in A(H).$$ 

\CCase{3} If $x,y$ belong to $V_i$ (possibly $x=y$) and $z \in V_j$, then $x',y' \in V_{i+1}$ and $z' \in V_{j+1}$. When $x<y$ and $x' <y'$, then by definition $g(x,y,z)g(x',y',z')=xx' \in A(H)$. 
Now suppose that $x <y $ and $y'<x'$.  Since $<$ is a min-ordering on $V_i,V_{i+1}$, we have $xy' \in A(H)$, and hence, $$g(x,y,z)g(x',y',z')=xy' \in A(H).$$  
By symmetry, the other remaining cases can be handled similarly. 
\end{proof}

\begin{theorem}
There is a $|V(H)|^2$-approximation algorithm for MinHOM($H$) when the target digraph $H$ admits a k-min-ordering, $k >1$. 
\end{theorem}
\begin{proof}
Let $V_0,V_1,\dots,V_{k-1}$ be a partition of the vertices of $H$ according to the k-min-ordering; i.e. every arc of $H$ is from a vertex in $V_{l}$ to a vertex in $V_{l+1}$, $0 \le k-1$ (sum module $k$).   Clearly a mapping $g : V(H)\to \overrightarrow{C_k}$ with 
$g(a)=l$ when $a \in V_{l}$, $ l \in \{0,1,\dots,k-1\}$, is a homomorphism from $H$ to $\overrightarrow{C_k}$

Let $D$ be the input digraph together with the costs. Observe that if $D$ is homomorphic to $H$, then $D$ must be homomorphic to $\overrightarrow{C_k}$. We may assume that $D$ is weakly connected. Otherwise, each weakly connected component of $D$ is treated separately.

Let $x$ be a fixed vertex in $D$ and let $\psi_{\ell}$ be a homomorphism from $D$ to $\overrightarrow{C_k}$ where $\psi_{\ell}(x)=\ell$, $\ell \in \{0,1,\dots,k-1\}$. We design an approximation algorithm for MinHOM($H$) in which $x$ is mapped to $V_{\ell}$ of $H$. In order to find the approximation algorithm for MinHOM($H$) for the given digraph $D$, we consider each homomorphism $\psi_{\ell}(x)=\ell$, $\ell \in \{0,1,\dots,k-1\}$ and find an approximation algorithm from $D$ to $H$ corresponding to $\psi_{\ell}$ and output the one with best performance.  
For simplicity of notations we work with $\phi=\psi_0$. Let $U_0,U_1, U_2,\dots,U_{k-1}$ be the partition of the vertices in $D$ under $\phi$, i.e. $\phi^{-1}(\ell)=U_{\ell}$.  

Consider the following LP with set of constraints called $\mathcal{S}_k$. For every $u \in U_{\ell}$ and every 
$a_i \in V_{\ell}$, $\ell \in \{0,1,\dots,k\}$ define a variable $0 \le u_i \le 1$. For every vertex $a_i \in V_{j}$, $j \in \{0,1,\dots,k-1\}$,
let $\ell^+(i)$ be the first index $i'$ such that $b_{i'} \in V_{j+1}$ in the ordering $<$ such that $a_ib_{i'} \in A(H)$ and let $\ell^{-1}(i)$ be the first $c_{r} \in V_{j-1}$ in the ordering $<$ such that $c_ra_i \in A(H)$.

\begin{tcolorbox}[colback=white]
\begin{center}
\scalebox{0.8}{
\begin{tabular}{r l l l}
\text{min} & $\sum\limits_{\substack{\ell \in \{0,1,\dots,k-1\} \\ v \in U_{\ell},i \in V_{\ell}}} c(v,a_i)(v_i-v_{i+1})$&&\\
\text{subject to:} &
 $v_i \geq 0$ & & \newtag{(A1)}{A1}\\
& $v_1  =  1$  & $\forall \ell \text{ and every }  v \in U_{\ell}, a_i \in V_{\ell}$ &  \newtag{(A2)}{A2}\\
& $v_{p+1}  =  0$  & $|V(H)|=p$ & \newtag{(A3)}{A3}\\
& $v_{i+1}  \leq  v_{i}$ & $\forall \ell \text{ and every }  v \in U_{\ell}, a_i \in V_{\ell}$ & \newtag{(A4)}{A4}\\
& $v_{i+1}  =  v_{i}$ & $\text{if } a_i \not\in L(v)$ & \newtag{(A5)}{A5}\\
& $u_{i}  \leq  v_{l^+(i)}$ & $\forall\text{ } uv \in A(D)$ & \newtag{(A6)}{A6}\\
& $v_{i}  \leq  u_{l^-(i)}$ & $\forall\text{ } uv \in A(D)$&\newtag{(A7)}{A7}
\end{tabular}}
\end{center}
\end{tcolorbox}

Let $a_1,a_2,\dots,a_p$ be the vertices in $V_{\ell}$ according to the  $k$-min-ordering $<$,  and let  
$b_1,b_2,\dots,b_q$ be the vertices in $V_{\ell+1}$ according to $<$. 

Let $E=A(H)$ and define $H'$ to be the digraph with vertex set $V(H)$ and arc set
$E \cup E'$. Here $E'$ is the set of arcs added into $A(H)$ so that the resulting digraph admit a $k$-min-max-ordering. Note that $E$ and $E'$ are disjoint sets. Let $E'_{\ell}$ denote the set of all the pairs $(a_i,b_j) \in V_{\ell} \times V_{\ell+1}$ such
that $a_ib_j$ is not an arc of $H$, but there is an arc $a_ib_{j'}$ of $H$ with $j' < j$ and an arc $a_{i'}b_j$
of $H$ with $i' < i$. Observe that $E'= \bigcup_{\ell=0}^{\ell=k-1} E'_{\ell}$.

For every arc $e = a_ia_j \in E'_{\ell}$ and every arc $uv \in A(D)$, $u \in U_{\ell}, v \in U_{\ell+1}$  three of the following set of  inequalities is added to ${\cal S}$ (i.e. either \ref{A8}, \ref{A11} or \ref{A9}, \ref{A10} or \ref{A9},\ref{A11}).


\begin{tcolorbox}[colback=white]
\begin{center}
\scalebox{0.8}{
\begin{tabular}{l l l l l}
$v_j \leq u_s$ & $+$ & $\sum\limits_{\substack{ a_t \in L(u)\\ a_tb_j \in E_{\ell} \\ t<i }}(u_t-u_{t+1})$ & \text{if $a_s$ is the first in-neighbor of $b_j$ after $a_i$} &\newtag{(A8)}{A8}\\ 
$v_j \leq v_{j+1}$ & $+$ & $\sum\limits_{\substack{ a_t \in L(u)\\ a_tb_j\in E_{\ell}\\ t<i}}(u_t-u_{t+1})$ & \text{if $b_j$ has no in-neighbor after $a_i$} &\newtag{(A9)}{A9}\\
$u_i \leq v_s$ & $+$ & $\sum\limits_{\substack{ b_t \in L(v) \\ a_ib_t\in E_{\ell}\\ t<j}}(v_t-v_{t+1})$ & \text{if $b_s$ is the first out-neighbor of $a_i$ after $b_j$} & \newtag{(A10)}{A10}\\
$u_i \leq u_{i+1}$ & $+$ & $\sum\limits_{\substack{ b_t \in L(v)\\ a_ib_t\in E_{\ell} \\ t<j}}(v_t-v_{t+1})$ & \text{if $a_i$ has no out-neighbor after $b_j$} & \newtag{(A11)}{A11}
\end{tabular}}
\end{center}
\end{tcolorbox}

Moreover, by pair consistency, we can add the following constraints for every $u \in U_{\ell}$ and every $v \in U_{\ell'}$ in $V(D)$ and $a_i \in L(u)$:
\begin{center}
\begin{tabular}{l l l l}
$u_i-u_{i+1}$ & $\le$ & $\sum\limits_{\substack{j: \\(a_i,b_j) \in L(u,v)}}(v_j-v_{j+1})$ &\newtag{(A12)}{A12}
\end{tabular}
\end{center}
Let the extended of $\mathcal{S}_k$ be denoted by $\widehat{\mathcal{S}_k}$. 

By an argument similar to that in the previous section, one can show the following that there is a one-to-one correspondence between the homomorphisms from $D$ to
$H$ and integer solutions to $\widehat{\mathcal{S}_k}$.

In what follows we outline the process of rounding the fractional solutions of the LP to obtain an integral solution, and hence, a homomorphism from $D$ to $H$ (see \ref{shift-extended-pic}). In the first stage of the algorithm, we use a random variable $X\in [0,1]$ and round the fractional values according to $X$. This means, if $u_i < X$ then $u'_i$ is set to zero, otherwise we set $u'_i=1$. 

The intention is to map $u \in D$ to the vertex $a_i$ of $H$ when $u'_i=1$ and $u'_{i+1}=0$. However, we may set $u'_{i}=v'_j=1$, $u'_{i+1}=v'_{j+1}=0$ where $u \in U_{\ell},v \in U_{\ell+1}$, $a_i \in V_{\ell}$, $b_j \in V_{\ell}+1$ and $a_ib_j \in E'_{\ell}$, i.e. $a_ib_j$ is not an arc of $H$ but it is one of the added arcs into $H$. In other words, what we have obtained would not be a homomorphism, and hence, we have to fix this partial integral assignment. To keep track of fixings, we may assume sum $i+j$ is maximum. 

We may assume that $b_j$ does not have any in-neighbor in $V_{\ell}$ after $a_i$. Now we use a random variable $Y\in [0,1]$ to select an out-neighbor $b_t \in V_{\ell+1}$ of $a_i$ before (in the ordering $<$) $b_j$ and shift the image of $v$ from $b_j$ to $b_t$. The vertex $b_t$ is selected according to random variable $Y$ with the same rule as the one described in Section \ref{sec-Approx-ALG} (see the description after Lemma \ref{set-index-shifting}). However, this could force us to shift the image of some out-neighbor of $v$, say $w \in V_{\ell+2}$ (subscript in modulo $k$). Therefore, we deploy a BFS search (applying a version of shift procedure in Algorithm \ref{alg:app-MinHom} ) to fix the images of the vertices of $D$ that may need to be changed because of the initial change in shifting the image of $v$ to $b_t$ (see the Figure \ref{shift-extended-pic}). 
We use the same strategy as used in the case of the min-ordering to round the values of $\widehat{\mathcal{S}_k}$ and obtain an integral solution. The calculation to obtain the approximation ratio is almost identical to ones in proof of Lemma \ref{pr-bound}. 
\end{proof}
\begin{figure}[t]
  \centering
  \begin{subfigure}{.35\linewidth}
    \includegraphics[width = \linewidth]{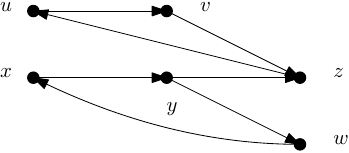}
    \caption{Input digraph $D$}
  \end{subfigure}
  \hspace{4em}
  \begin{subfigure}{.35\linewidth}
    \includegraphics[width = \linewidth]{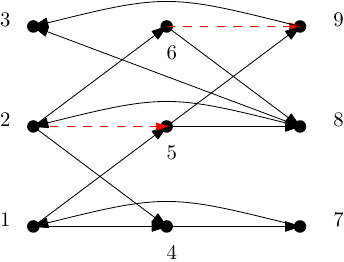}
    \caption{Target digraph $H$}
  \end{subfigure}
  \caption{An illustration of the algorithm for $k$-min-ordering. In digraph $H$, $1,2,3,4,5,6,7,8,9$ is a $3$-min-ordering. The dash arcs are the missing arcs. Suppose after the first step of rounding $u'_2=v'_5=z'_8=1$, $v'_6=u'_3=z'_9=0$.   Then the algorithm shifts the image of $v$ from $5$ to $4$ and consequently the image of $z$ from $8$ to $7$.}
  \label{shift-extended-pic}
\end{figure}

\section{A Dichotomy for Graphs}\label{sec-majority}

Feder and Vardi~\cite{feder1998computational} proved that if a graph $H$ admits a conservative majority polymorphism, then LHOM(H) is polynomial-time solvable. Later, Feder \emph{et al.,}~\cite{feder2003bi} showed that LHOM(H) is polynomial-time solvable if and only if $H$ is a \textit{bi-arc} graph. Bi-arc graphs are defined as follows.

Let $C$ be a circle with two specified points $p$ and $q$ on $C$. A bi-arc is an ordered pair of arcs $(N,S)$ on $C$ such that $N$ contains $p$ but not $q$, and $S$ contains $q$ but not $p$. A graph $H$ is a bi-arc graph if there is a family of bi-arcs $\{(N_x, S_x) : x \in V(H)\}$ such that, for any $x, y \in V(H)$, not necessarily distinct, the following hold:
\begin{itemize}
\item[--] if $x$ and $y$ are adjacent, then neither $N_x$ intersects $S_y$ nor $N_y$ intersects $S_x$;
\item[--] if $x$ and $y$ are not adjacent, then both $N_x$ intersects $S_y$ and $N_y$ intersects $S_x$.
\end{itemize}

We shall refer to $\{(N_x, S_x) : x \in V(H)\}$ as a bi-arc representation of $H$. Note that a bi-arc representation cannot contain bi-arcs $(N , S), (N',S')$ such that $N$ intersects $S'$ but $S$ does not intersect $N'$ and vice versa. Furthermore, by the above definition a vertex may have a self loop. 
\begin{theorem}[\cite{brewster2008near,feder2003bi}]
A graph admits a conservative majority polymorphism if and only if it is a bi-arc graph.
\end{theorem}

\begin{definition}[$G^*$]
\label{construct-G*}
Let $G=(V,E)$ be a graph. Let $G^*$ be a bipartite graph with partite sets $V,V'$ where $V'$ is a copy of $V$. Two vertices $u \in V$, and $v' \in V'$ of $G^*$ are adjacent in $G^*$ if and only if $uv$ is an edge of $G$.  
\end{definition}

A \textit{circular arc} graph is a graph that is the intersection
graph of a family of arcs on a circle. We interpret the concept of an intersection
graph literally, thus any intersection graph is automatically reflexive (i.e. there is a loop at every vertex), since a set always intersects itself. A bipartite graph whose complement is a circular arc
graph, is called a \textit{co-circular arc} graph. 
Note that co-circular arc graphs are irreflexive, meaning that no vertex has a loop. 
Let $G=(X,Y,E)$ be a bipartite graph with partite sets $X$ and $Y$, and edge set $E$. Recall that we say
$G$ admits a min-ordering, if there is an ordering $x_1,x_2,\dots,x_p$ and ordering 
$y_1,y_2,\dots,y_q$ so that when $x_iy_j \in E$ and $x_ry_s \in E$ with $i<r$ and $s<j$ then 
$x_ry_s$ is an edge of $G$. With this definition, if we oriented all the edges of $E$ from 
$X$ to $Y$, then we obtain digraph $D$ for with it is easy to see that the ordering 
$x_1,x_2,\dots,x_p, y_1,y_2,\dots,y_q$ is a min-ordering. 

\begin{lemma}\label{co-circular}
Let $H^*$ be the bipartite graph constructed from a bi-arc graph $H$, according to Definition~\ref{construct-G*}. Then the following hold.  
\begin{itemize}
    \item[--]  $H^*$ is a co-circular arc graph.  
    \item[--] $H^*$ admits a min-ordering. 
\end{itemize}
\end{lemma}
\begin{proof}
It is easy to see that $H^*$ is a co-circular arc graph. From a bi-arc representation $\{(N_i,S_i) : i \in V(H) \}$ of $H$, we obtain a co-circular arc representation of $H^*$ by choosing, for $i \in H$, the arc $N_i$ for vertex $i \in H^*$ and the arc $S_i$ for vertex $i' \in H^*$. A bipartite graph admits a min-ordering if and only if it is  co-circular arc graph~\cite{arash-esa}.  $H^*$ is a co-circular arc graph, and hence, it admits a min-ordering.  
\end{proof}

Let $H$ be a bi-arc graph, with vertex set $I$, and let $H^*$ be the bipartite graph constructed from $H$ having vertices $(I,I')$ according to Definition~\ref{construct-G*}. Let $a_1,a_2,\dots,a_p$ be an ordering of the vertices in $I$ and $b_1,b_2,\dots,b_p$ be an ordering of the vertices of $I'$. Note that each $a_i$ has a copy $b_{\pi(i)}$ 
in $\{b_1,b_2,\dots,b_n\}$ where $\pi$ is a permutation on $\{1,2,3,\dots,p\}$. By Lemma \ref{co-circular}, let us assume $a_1,a_2,\dots,a_p, b_1,b_2, \dots,b_p$ is a min-ordering for $H^*$.  

Let $G$ be the input graph with vertex set $V$ and let $c$ be a given cost function. 
Construct $G^*$ from $G$ with vertex set $V\cup V'$ as in Definition~\ref{construct-G*}. Now construct an instance of the MinHOM($H^*$) for the input graph $G^*$ and set $c(v',b_{\pi(i)})=c(v,a_i)$ for $v \in V$, $v' \in V'$. 


\begin{lemma}
There exists a homomorphism $f:G\to H$ with cost $\mathfrak{C}$ if and only if there exists homomorphism $f^*:G^*\to H^*$ with cost $2\mathfrak{C}$ such that, if $f^*(v)=a_i$ then $f^*(v')=b_j$ with $j = \pi(i)$.
\end{lemma}

We first perform the arc-consistency and pair-consistency procedures for the vertices in $G^*$. Note that if $L(u)$ contains the element $a_i$ then $L(u')$ contains $b_{\pi(i)}$ and when $L(u')$ contains some $b_j$ then $L(u)$ contains $a_{\pi^{-1}(j)}$. Next, we define the system of linear equations $\widehat{S^*}$ with the same construction as in Sections \ref{sec-minmax}, \ref{min-to-max}. Equivalently, one can use the LP formulation in \cite{arash-esa}. However, for the sake of completeness we present the entire LP in this section. 

Consider the following linear program. For every vertex $v \in V$ from $V(G^*)=(V,V')$ and every vertex $a_i \in I$ 
from $V(H^*)=(I,I') $ define a variable $v_{i}$.  For every vertex $v' \in V'$ from $V(G^*)$ and every vertex $b_i \in I'$ 
from $V(H^*)$ define a variable $v'_{i}$. We also define the variables $v_{p+1},v'_{p+1}$ for every $v \in V$ whose value is set to zero.
Now the goal is to solve the following linear program :
\begin{tcolorbox}[colback=white]
\begin{center}
\scalebox{0.77}{
\begin{tabular}{r l l l}
\text{min} & $\sum\limits_{v,i} c(v,a_i)(v_i-v_{i+1})+\sum\limits_{v',j}c(v',b_j)(v'_j - v'_{j+1})$\\
\text{subject to:}
&$v_i,v'_{\pi(i)} \geq 0$ & & \newtag{(CM1)}{CM1}\\
& $v_1  = v'_1= 1$ & & \newtag{(CM2)}{CM2}\\
& $v_{p+1}  =v'_{p+1}=  0$ & & \newtag{(CM3)}{CM3}\\
& $v_{i+1}  \leq  v_{i}$  \ \ and \ \ $v'_{\pi(i)+1}  \leq  v'_{\pi(i)}$ & & \newtag{(CM4)}{CM4}\\
& $v_{i+1}  =  v_{i}$ \ \ and \ \ $v'_{\pi(i)+1}= v'_{\pi(i)}$ &  $\text{if } a_i \not\in L(v)$ & \newtag{(CM5)}{CM5}\\
& $u_{i}  \leq  v'_{l^+(i)}$  & $\forall\text{ } uv' \in E(G^*)$ & \newtag{(CM6)}{CM6}\\
& $v'_{i}  \leq  u_{l^-(i)}$  & $\forall\text{ } uv' \in E(G^*)$ &
\newtag{(CM7)}{CM7}\\
& $u_i-u_{i+1}=u'_{\pi(i)} -u'_{\pi(i)+1}$  & $\forall u,u' \in G^* ,\forall a_i,b_{\pi(i)} \in H^*$ & \newtag{(CM8)}{CM8}
\end{tabular}}
\end{center}
\end{tcolorbox}

Here $l^+(i)$ is the first index $j$, such that $a_ib_j$ is an edge of $H^*$, and $l^{-1}(i)$ is the first index $j$ such that $a_jb_i$ is an edge of $H^*$. 

Let $E'$ denote the set of all the pairs $(a_i,b_j)$ such that $a_ib_j$
is not an edge of $H^*$, but there is an edge $a_ib_{j'}$ of $H^*$ with $j' < j$ and an edge $a_{i'}b_j$
of $H^*$ with $i' < i$. Let $E=A(H^*)$ and define $H'^*$ to be the digraph with vertex set $V(H^*)$ and edge set
$E \cup E'$. Note that $E$ and $E'$ are disjoint sets. For every edge $e = a_ib_j \in E'$ 
and every edge $uv \in E(G^*)$, by Observation \ref{observation2}, three of the following set of inequalities will be added to $\widehat{S^*}$ (i.e. either \ref{CM9}, \ref{CM12} or \ref{CM10}, \ref{CM11} or \ref{CM10}, \ref{CM12}).

\begin{tcolorbox}[colback=white]
\begin{center}
\scalebox{0.8}{
\begin{tabular}{l l l l l}
 $v'_j \leq u_s$ & $+$ & $\sum\limits_{\substack{ a_t \in L(u)\\ a_tb_j \in E \\ t<i }}(u_t-u_{t+1})$ & \text{if $a_s$ is the first in-neighbor of $b_j$ after $a_i$} & \newtag{(CM9)}{CM9}\\ 
$v'_j \leq v'_{j+1}$ & $+$ & $\sum\limits_{\substack{ a_t \in L(u)\\ a_tb_j\in E \\ t<i}}(u_t-u_{t+1})$ & \text{if $b_j$ has no in-neighbor after $a_i$} & \newtag{(CM10)}{CM10}\\
$u_i \leq v'_s$ & $+$ & $\sum\limits_{\substack{ b_t \in L(v')\\ a_ib_t \in E\\ t<j}}(v'_t-v'_{t+1})$ & \text{if $b_s$ is the first out-neighbor of $a_i$ after $b_j$} & \newtag{(CM11)}{CM11}\\
 $u_i \leq u_{i+1}$ & $+$ & $\sum\limits_{\substack{ b_t \in L(v')\\ a_ib_t\in E\\ t<j}}(v'_t-v'_{t+1})$ & \text{if $a_i$ has no out-neighbor after $b_j$} &  \newtag{(CM12)}{CM12}
\end{tabular}}
\end{center}
\end{tcolorbox}

\begin{lemma}\label{one-to-one-majority-LP}
If $H$ is a bi-arc graph, then there is a one-to-one correspondence between homomorphisms from $G$ to
$H$ and integer solutions of $\widehat{S^*}$.
\end{lemma}
\begin{proof} For a homomorphism $f: G \to H$, if $f(v)=a_t$ we set $v_i=1$ 
for all $i\leq t$, otherwise, we set $v_i=0$, we also set $v'_{j}=1$ for all $j \leq \pi(i)$ 
and $v'_{j+1}=0$ where $\pi(i)=j$. We set $v_1=1$, $v'_1=1$  and $v_{p+1}=v'_{p+1}=0$ for all $v,v' \in V(G^{*})$. Now all the variables are non-negative and we have $v_{i+1}\leq v_i$ and $v'_{j+1} \le v'_{j}$. 
Note that constraint~\ref{CM8} is satisfied  by this assignment. 
We first show that $u_i \le v'_{l^+(i)}$ for every edge $uv' \in E(G^*)$
Suppose for contradiction that $u_i=1$ and $v'_{l^+(i)}=0$ and let $f(u)=a_r$ and $f(v)=a_s$. This implies that $u_r=1$, whence $i\leq r$; and $v'_s=1$, whence $s <l^+(i)$. Since both $a_ib_{l^+(i)}$ and $a_rb_s$  are edges of $H^*$ with $i \leq r$ and $s<l^+(i)$, the fact that $H^*$ has a min-ordering implies that $a_ib_s$ must also be an edge of $H^*$,
contradicting the definition of $l^+(i)$. The proof for $v'_{j}  \leq  u_{l^-(i)}$ is analogous. Therefore, constraints 
\ref{CM6} and \ref{CM7} are 
satisfied. It is also easy to see that (similar to the case for digraphs) 
that constraints \ref{CM9},~\ref{CM10},~\ref{CM11}, and \ref{CM12} also satisfied by this assignment.

Conversely, suppose there is an integer solution for $\widehat{S^*}$. 
First we define a homomorphism $g : G^* \rightarrow H^*$ as follows : let $g(u) = a_i$ where $i$ is the largest subscript with $v_i = 1$, and $g(v')=b_j$ when $j$ is the largest subscript with $v_j = 1$. We prove that
this is indeed a homomorphism by showing that every edge of $G^*$ is mapped to an edge of $H^*$. Let $uv'$ be an edge of $G^*$ and assume $g(u)=a_r$, $g(v')=b_s$ 
We show that $a_rb_s$ is an edge in $H^*$. Observe that, by~\ref{CM6} and~\ref{CM7}, $1=u_r\leq v'_{l^+(r)}\leq 1$ and $1=v'_s \leq u_{l^-(s)}\leq 1$, therefore we must have $v'_{l^+(r)}=u_{l^-(s)}=1$. 
Since $r$ and $s$ are the largest subscripts such that $u_r=v'_s=1$ then $l^+(r)\leq s$ and $l^-(s)\leq r$. Since $a_rb_{l^+(r)}$ and $a_{l^-(s)}b_s$ are edges of $H^*$, we must have the edge $a_rb_s$ in $H^*$ because $H^*$ admits a min-ordering.  Furthermore, $g(u) = a_i$ if and only if $u_i = 1$ and $u_{i+1} = 0$,
so, $c(u, a_i)$
contributes to the sum if and only if $g(u) = a_i$ and $c(v',b_j)$ contributes to the sum if and only if $g(v')=b_j$.

Now let $f(u)=a_i$ when $g(u)=a_i$. We show that if $uv$ is an edge 
of $G$ then $f(u)f(v)$ is an edge of $H$. Since $g$ is a homomorphism 
from $G^*$ to $H^*$,  $g(u)g(v') \in E(H^*).$ Suppose $g(v')=b_j$. 
This implies $u_i=v'_{j}=1$ and $u_{i+1}=v'_{j+1}=0$. 
Now by constraint~\ref{CM12}, we have $v_{\pi^{-1}(j)}=1$, 
and $v_{\pi^{-1}(j)+1}=0$, and hence, we have $f(v)=a_{\pi^{-1}(j)}$. 
Now by definition of $H^*$, $a_ia_{\pi^{-1}(j)}$ is an edge of $H$ because $a_ib_j$ is an edge of $H^*$. 
Furthermore, $f(u) = a_i$ if and only if $u_i = 1$ and $u_{i+1} = 0$, so, $c(u, a_i)$ contributes to the sum if and only if $f(u) = a_i$. 
\end{proof}

Once again we round an optimal fractional solution of $\widehat{S^*}$, using a random variable $X\in[0,1]$. Let $f$ be a mapping from $V(G^*)$ to $V(H^*)$ obtained by rounding. We give an algorithm that modifies $f$ so that $f:G \to H $ is a homomorphism (i.e. an integral solution that satisfies $\widehat{S^*}$).

\begin{algorithm}[t]
\caption{Approximation MinHOM($H$) for graphs}\label{alg:graph-approx}
\begin{algorithmic}[1]
\Procedure{Approx--Graph-MinHOM}{$H$}
\State Construct $H^*$, $G^*$ from $H$, $G$ respectively, as in Definition \ref{construct-G*}
\State Perform rounding and Stage 1 on fractional values returned by solving LP  $\widehat{S^*}$. \Comment{Similar as calling Algorithm~\ref{alg:app-MinHom} on $H^*$, $G^*$, and $\widehat{S^*}$.} 
 
\State Let $f$ be the homomorphism from $G^*$ to $H^*$ returned in the previous step
\State $f=$\textsc{Shift}$(f)$

\State \textbf{return} $f$\Comment{$f$ is a homomorphism from $G$ to $H$}
\EndProcedure

\end{algorithmic}
\end{algorithm}


\begin{theorem}\label{shifting-majority}
There exists a randomized algorithm that modifies $f$ and obtain a homomorphism from $G$ to $H$. Moreover, the expected cost of the homomorphism returned by this algorithm is at most $2|V(H)| \cdot {OPT}$. 
\end{theorem}
\begin{proof}

For every variable $u_i$, $u \in V(G^*)$, set $\hat{u}_i=1$ if $X \le u_i$  else $\hat{u_i}=0$.
Similarly for every $v'_j$, $v' \in V(G^*)$, set $\hat{v}'_j=1$ if $X \le v'_j$  else $\hat{v}'_i=0$. The algorithm has two stages after rounding the fractional solution using the random variable $X$. 

\medskip

\noindent \textbf{Stage 1. Fixing the edges  $uv'$ of $G^*$ that have been mapped to non-edges $a_ib_j$ of $H^*$:} 
Suppose for some edge $uv'$ of $G^*$, $\hat{u}_i=1$, $\hat{u}_{i+1}=0$, $\hat{v}'_j=1$, $\hat{v}'_{j+1}=0$. By Observation~\ref{observation2}, either $b_{j}$ has no in-neighbor after $a_i$ or $a_i$ has no out-neighbor after $b_{j}$. Suppose the former is the case. 
We also note that because of the constrains \ref{CM5}, \ref{CM6}, $a_ib_j$ is one of the edges that should be added into $H^*$ in order to obtain a min-max-ordering for $H'^{*}$.  
 Suppose, for edge $uv'$ of $G^*$, $f(u)=a_i,f(v)=a_j$ where $a_ia_j \in E'$; i.e. $a_ib_{\pi(j)} \not\in A(H^*)$. We may assume that $a_ia_j$ is the last such non-edge in $H$ ($i+j$ is maximum) when we look at the min-ordering of $H^*$.

Choose a random variable $Y \in [0,1]$, which will guide us to shift the image of $v'$ from $b_j$ to some $b_{t}$ where $a_ib_t\in E$, and $b_t$ appears before $b_{j}$ in the min-ordering of $H^*$. Consider the set of such $b_t$s ( by definition of the min-ordering of $H^*$, this set is non-empty), and suppose it consists of $b_t$ with subscripts $t$ ordered as $t_1 < t_2 < \dots t_k$. Let $P_{v',t} = \frac{v'_t - v'_{t+1}}{P_{v'}}$
with $P_{v'} = \sum\limits_{\small a_ib_t\in E(H^*),~t<j} (v'_t - v'_{t+1})$. Select the vertex $b_{t_q}$ if
$\sum\limits_{\small p=1}^{\small q} P_{v',{t_p}} < Y \leq \sum\limits_{\small p=1}^{\small q+1} P_{v',{t_p}}$.
Thus, a concrete $b_t$ is selected with probability $P_{v',t}$, which is proportional to the difference
of the fractional values $v'_t-v'_{t+1}$. Observe that there is no need to shift the image of some vertex $w$ which is an in-neighbor of $v'$ from its current value to some other vertex (because of shifting the image of $v$). 

Now we note that the probability of shifting the image of some $v'$ from $b_j$ to $b_t$ is at most $v'_t-v'_{t+1}$. Note that as long as such edges $uv'$ exists, we repeat the shifting procedure.  
At the end of this stage we have obtained a homomorphism $f$ from $G^*$ to $H^*$.

\medskip

\begin{algorithm}
\caption{The shifting procedure for unstable vertices (Stage 2)}\label{alg:graph-shift}
\begin{algorithmic}

\Procedure{Shift}{$f$}
\While {there are unstable vertices }
    \State Let $u$ be a vertex with $f(u)=a_i$ and $f(u') \ne b_{\pi(i)}$ where $i$ is maximum.

    \State Let Q be a Queue. $Q.enqueue(u')$ 

    \While { $Q$ is not empty}

        \State $x \gets Q.dequeue( )$
        \If{ $x=v'$ }
        \State $f(v')  \gets b_{\pi(i)}$ where $f(v)=a_i$.
        \For{ $wv' \in E(D)$ with $a_{\ell}=f(w)$ and $f(w') \ne b_{\pi(\ell)}$}
            \State $Q.enqueue(w)$  
        \EndFor 
        \ElsIf{ $x=v$}
            \State $f(v)  \gets a_{i}$ where $f(v')=b_{\pi(i)}$.
            \For{$vw' \in E(D)$ with $a_{\ell}=f(w)$ and $f(w') \ne b_{\pi(\ell)}$}
                \State $Q.enqueue(w')$  
            \EndFor 
        \EndIf 
    \EndWhile
\EndWhile 
\State \textbf{return} $f$\Comment{$f$ is a homomorphism from $G$ to $H$}
\EndProcedure
\end{algorithmic}
\end{algorithm}

\noindent  \textbf{Stage 2. Making the assignment consistent with respect to both orderings:} We say a vertex $u \in V$ of $V(G^*)=(V,V')$ is {\em unstable} if $f(u)=a_i$, $f(u')=b_q$ where $q \ne \pi(i)$. 
Now we start a BFS in $V(G^*)$ and continue as long as there exists an unstable vertex $u$ in $G^*$. At each step, we start from the greatest subscripts $i$ for which there exists an unstable $u$ with $f(u)=a_i$. During the BFS, one of the following is performed:

\begin{enumerate}
    
    \item \label{action1} shift the image of $u'$ from $b_q$ to $b_{\pi(i)}$.
    \item \label{action2} shift the image of $u$ from $a_i$ to $a_{\pi^{-1}(q)}$. 
\end{enumerate}

As a consequence of the above actions we would have the following cases:

\CCase{1} We change the image of $u'$ from $b_q$ to $b_{\pi(i)}$ (with $f(u)=a_i$), and there exists some $v' \in V'$ such that $uv' \in E(G^*)$ with $f(v)=a_j$ and $f(v')=b_{\pi(j)}$.

We note that $a_ib_{\pi(j)}$ is an edge because $uv'$ is an edge, and hence, $a_j b_{\pi(i)}$ is an edge of $H^*$. This would mean there is no need to shift the image of $v$ from $a_j$ to something else (see the Figure~\ref{fig:shifting-majority-1}). 
\medskip

\CCase{2} We change the image of $u'$ from $b_q$ to $b_{\pi(i)}$ (with 
$f(u)=a_i$), and there exists some edge $vu'$ of $H^*$ with $f(v)=a_j$ and $f(v')=b_{\ell}$ with $\ell \ne \pi(j)$. 

Such vertex $v$ is added into the queue, and once we retrieve $v$ from the queue we do the following: 
moving the image of $v$ from $a_{j}$ to $a_{\pi^{-1}(\ell)}$ (see the Figure \ref{fig:shifting-majority-2}).

Note that $a_ib_{\ell} \in E(H^*)$ because $vu'$ is an edge of $G^*$, and hence $a_{\pi^{-1}(\ell)} b_{\pi(i)} $ is an edge of $H^*$.

\medskip
\CCase{3} We change the image of $v$ from $a_j$ to some  $a_{\pi^{-1}(\ell)}$  (with $f(v')=b_{\pi(\ell)}$) 
and there exists some $vw'$ such that $f(w)=a_t$ and $f(w')=b_{\pi(t)}$. We note that $a_{t}b_{\ell} \in E(H^*)$ because $v'w$ is an edge, and hence, $a_{\pi^{-1}(\ell)} b_{r}$  is an edge of $H^*$. This would mean there is no need to shift the image of $w'$ to something else.
\medskip

\CCase{4} We change the image of $v$ from $a_j$ to some $a_{\pi^{-1}(\ell)}$ (with $f(v')=b_{\ell}$).
Let $r$ be a greatest subscript such that there exists some $vw'$ where $f(w)=a_t$ and $f(w')=b_r$
with $r \ne \pi(t)$, $t <i$. Such vertex $w'$ is added into the queue, and once we retrieve $w'$ from the queue we do the following: moving the image of $w'$ from $b_{r}$ to $b_{\pi^{-1}(t)}$. 

Note that $a_tb_{\ell} \in E(H^*)$ because $wv'$ is also an edge of $G^*$. Hence, $a_{\pi^{-1}(\ell)} b_{\pi^{-1}(t)}$ is an edge of $H^*$.

When Case~2 occurs, we continue the shifting. This would mean we may need to shift the image of some -neighbor $w'$ of $v$ accordingly. We continue the BFS from $v$, and modify the images of neighbors of $v$, say $w'$, to be consistent with new image of $v$. This means we encounter either Case~3 or Case~4. 
Suppose $f(w')=b_{t}$ or $f(w')=b_{\pi(t)}$
Then 
there is no need to change the image of $w'$. 
Otherwise, we change the image of $w'$ from $b_t$ to $b_{j}$ where $a_{\pi^{-1}(\ell)} b_{j}$ is an edge of $H^*$ and we need to consider Cases 3,4 for the current vertex $w$. When we are in Case 4, then consider Cases 1,2 and proceed accordingly. 

Note that during the BFS, if we encounter a vertex $x$ (or  $x'$) that has been visited before, then we would be at Case 1 or 3 and hence, no further action is needed for in-neighbors (out-neighbors) of $x$. We also note that at each step an unstable vertex $y$ is associated to some $a_{\ell}$ where $\ell$ is decreasing. Therefore, this procedure would eventually stop, and we will no longer have unstable vertices in $V$. 
\begin{figure}[t]
  \centering
  \begin{subfigure}{0.4\textwidth}
    \includegraphics[width=5cm,height=2.7cm]{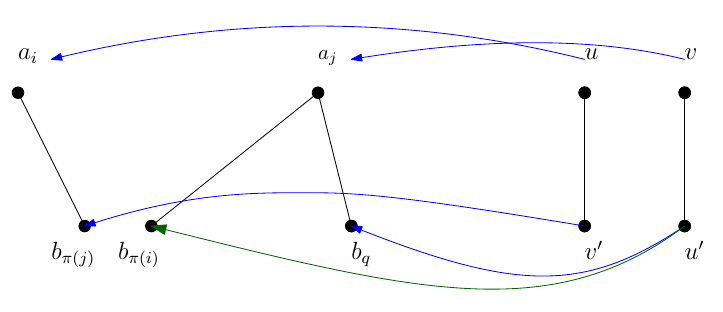}
    \caption{Case 1}
    \label{fig:shifting-majority-1}
  \end{subfigure}
  \hspace{2em}
  \begin{subfigure}{0.4\textwidth}
    \includegraphics[width=5cm,height=2.7cm]{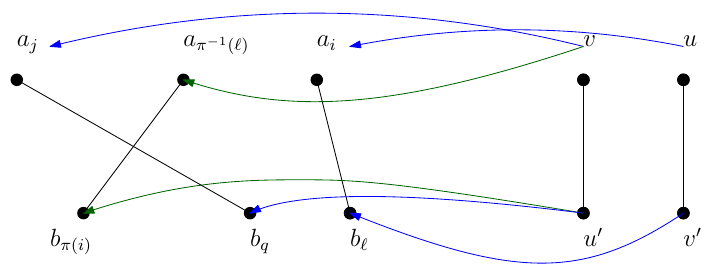}
    \caption{Case 2}
    \label{fig:shifting-majority-2}
    \hfill
  \end{subfigure}
  \caption{Illustrating the shifting process in Stage 2 of the algorithm. }\label{fig:shifting-majority}
\end{figure}

\paragraph{\textbf{Estimating the ratio}}
Vertex $v$ ($v'$, resp.) is mapped to $a_t$ ($b_t$, resp.) in three situations. 
\noindent {The first scenario} is where $v$ is mapped to $a_t$ by rounding (according to random variable $X$ in Stage 1) and is not shifted away. In other words, we have
$\hat{v}_{t} = 1$ and $\hat{v}_{t+1} = 0$ (i.e. $v_{t+1} \le X <v_t$) and these values do not change by the shifting procedure. 
Hence, for this case we have:
\begin{align*}
\P[f(v)=a_t] &=\P[v_{t+1} < X \leq v_t] \\ 
&\leq v_t-v_{t+1}
\end{align*}
Whence this situation occurs with probability at most $v_t-v_{t+1}$, and the expected contribution is at most $c(v,a_t)(v_t-v_{t+1})$.

\noindent {The second scenario} is where $f(v)$ is set to $a_t$ according to the random variable $Y$ in Stage 1. 

In this case $v$ is first mapped to $a_j, j > t$, by rounding according to variable $X$ and then re-mapped to $a_t$ during the shifting according to variable $Y$. We first compute the expected contribution for a fixed $j$, that is the contribution of shifting $v$ from a fixed $a_j$ to $a_t$. 

This happens if there exist $i$ and $u' \in V(H^*)$ such that $vu'$ is an edge of $G^*$ mapped to $a_jb_i \in E'$,
and then the image of $v$ is shifted to $a_t$ ($a_t<a_j$ in the min-ordering), where $a_tb_i \in E=E(H^*)$. In other words, we have
$\hat{u}'_{i} = \hat{v}_{j} = 1$ and $\hat{u}'_{i+1} = \hat{v}_{j+1} = 0$ after rounding; and then $v$ is shifted from $a_j$ to $a_t$. Therefore,
\begin{align*}
\P[\hat{u}_i'=\hat{v}_j=1 , \hat{u}'_{i+1}=\hat{v}_{j+1}=0]
&=\P[\max \{u'_{i+1},v_{j+1}\} < X \leq \min \{u'_i,v_j \}]\\
&= \min \{u'_i,v_j \} - \max \{u'_{i+1},v_{j+1}\}\\
&\leq v_j - v_{j+1}\\
&\leq \sum \limits_{\substack{t < j \\ a_tb_i \in E \\ a_t\in L(v)}} (v_{t} - v_{t+1}) \\
&= P_v 
\end{align*}
The last inequality is because $a_j$ has no out-neighbor after $b_i$ and it follows from inequality~\ref{CM9}. Having $vu'$ mapped to $a_jb_i$ in the rounding step, we shift $v$ to $a_t$ with probability $P_{v,t}={\frac{(v_t - v_{t+1})}{P_v}}$. Note that the upper bound $P_v$ is independent from the choice of $u$ and $b_i$. Therefore, for a fixed $a_j$, the probability that $v$ is shifted from $a_j$ to $a_t$ is at most $ \frac{v_t - v_{t+1}}{P_v}\cdot P_v = v_t - v_{t+1}$. There are at most $|V(H)|$ of such $b_i$'s, (causing the shift to $a_j$) and hence, the expected contribution of $v_t-v_{t+1}$ to the objective function is at most $|V(H)|c(v,t)(v_t-v_{t+1})$.

\noindent{The third scenario} is when the image of $v$ is shifted from some $a_j$ to $a_t$ in the second Stage of the shifting . More precisely, when one of the actions \ref{action1},\ref{action2} occurs. 

This happens because the image of $v'$ has been shifted from $b_q$ to $b_{\pi(t)}$ in Stage 2 according to variables $X$ or $Y$ (i.e. BFS). As we argued, in the previous scenarios, the overall expected value of shifting $v'$ from $b_q$ to $b_{\pi(t)}$ is $|V(H)|c(v,t)(v'_{\pi(t)}-v'_{\pi(t)+1})$. Since $v_t-v_{t+1}=v'_{\pi(t)}-v'_{\pi(t)+1}$, the overall expected value of shifting $v$ to $a_t$ is $|V(H)|(v_t-v_{t+1})$.
In conclusion, 
the expected contribution of $v_t-v_{t+1}$ to the objective function is $2|V(H)|c(v,t)(v_t-v_{t+1})$. 
\end{proof}

We remark that, as in the proof of Theorem~\ref{de-random}, the above algorithm can be de-randomized. By Lemma \ref{co-circular} and Theorem \ref{shifting-majority} we obtain the following classification theorem. 

\begin{theorem} \label{majority-dichotomy}
If $H$ admits a conservative majority polymorphism, then MinHOM(H) has a (deterministic) $2|V(H)|$-approximation algorithm, otherwise,  MinHOM(H) is inapproximable unless P=NP.
\end{theorem}

\section{Experiments}\label{sec-experiment}
\subsection{Finding a solution using GNU GLPK}

GLPK extends for GNU Linear Programming Kit, and it is an open source software package, written in C. It is intended for solving large-scale linear programming problems(LP). GLPK is a well-designed algorithm to solve LP problems, at a reasonable time. It implements different algorithms, such as the simplex method and the Interior-point method for non-integer problems and branch-and-bound together with Gomory's mixed integer cuts for integer problems. With GLPK we can add each constraint of our problem as a new row of a matrix. Before calculating the minimum cost, we have to set the type of solution we are looking for, integral only or if we allow a continuous solution.

\subsection{Experimental Results}

For our experiments, we have used graphs from four different classes namely, digraphs with a majority polymorphism, \emph{balanced} digraphs with a min-ordering, bipartite digraphs with a min-ordering, and DAT-free digraphs.
For each class, we have used a variety of target digraphs and sizes, ranging from $7$ to $15$. For a particular digraph in each class, we use a variety of input digraphs $D$, created randomly, with size from $100$ to $3000$. The cost of mapping an edge from digraph $D$ to an edge in digraph $H$ is randomly assigned, with values ranging from $5$ to $100000$. 
For each instance of MinHOM($H$) with input digraph $D$, we run our program twice, once for finding optimal fractional solution, and once for an integral solution. To calculate the ratio, for a single digraph $H$ of size $N$, we run our algorithm for each digraph $D$ of size $T$, $100$ times. We then get the ratio by calculating the average of fractional solution and integral solution, for every instance of different sizes of $D$. The target digraphs that we have examined are given next to the charts. All of the experiments indicate a very small integrality gap.

\begin{figure}[H]
  \centering
  \begin{subfigure}{.6\linewidth}
    \includegraphics[width =
    \linewidth]{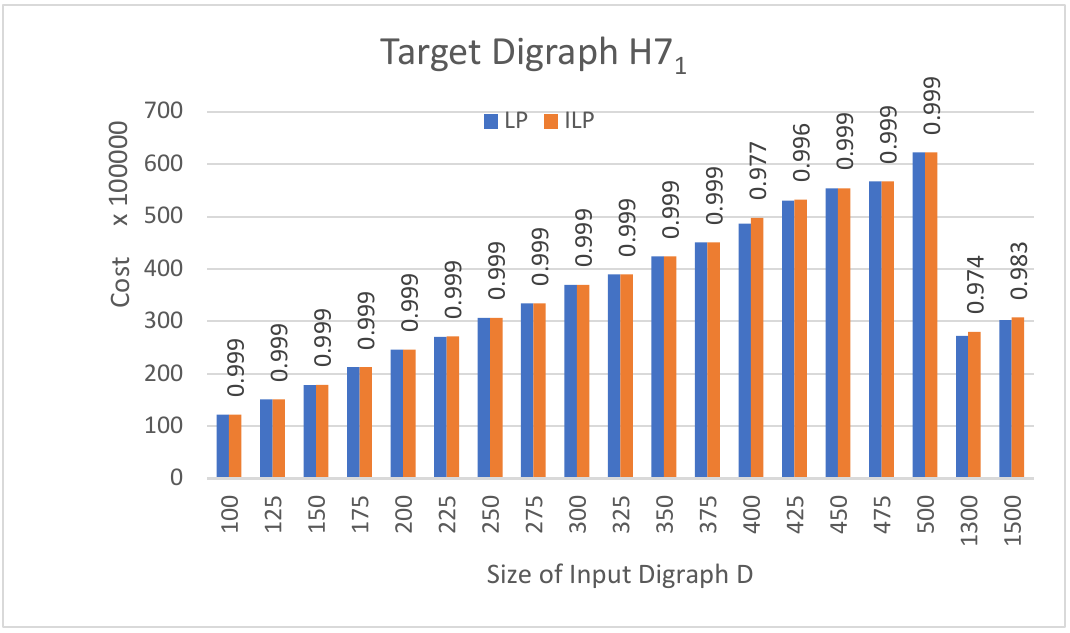}
  \end{subfigure}
  \hspace{2em}
  \begin{subfigure}{.2\linewidth}
    \includegraphics[scale=0.6]{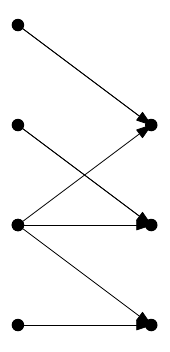}
  \end{subfigure}
  \caption{Comparison between LP and ILP when the target digraph is $H7_1$. $H7_1$ is a bipartite graph that admits a min-ordering.}
  \label{examples-H71}
\end{figure}

\begin{figure}[H]
  \centering
  \begin{subfigure}{.6\linewidth}
    \includegraphics[width =
    \linewidth]{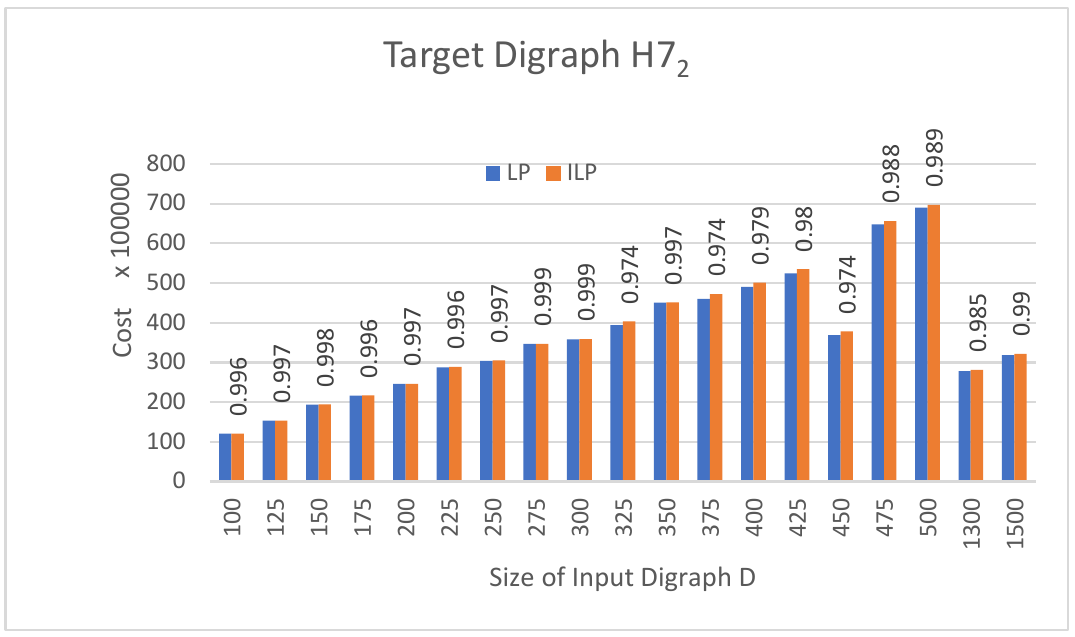}
  \end{subfigure}
  \hspace{2em}
  \begin{subfigure}{.2\linewidth}
    \includegraphics[scale=0.6]{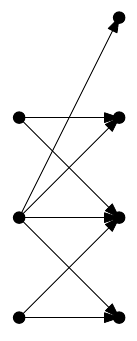}
  \end{subfigure}
  \caption{Comparison between LP and ILP when the target digraph is $H7_2$. $H7_2$ is a bipartite graph that admits a min-ordering.}
  \label{examples-H72}
\end{figure}

\begin{figure}[H]
  \centering
  \begin{subfigure}{.6\linewidth}
    \includegraphics[width =
    \linewidth]{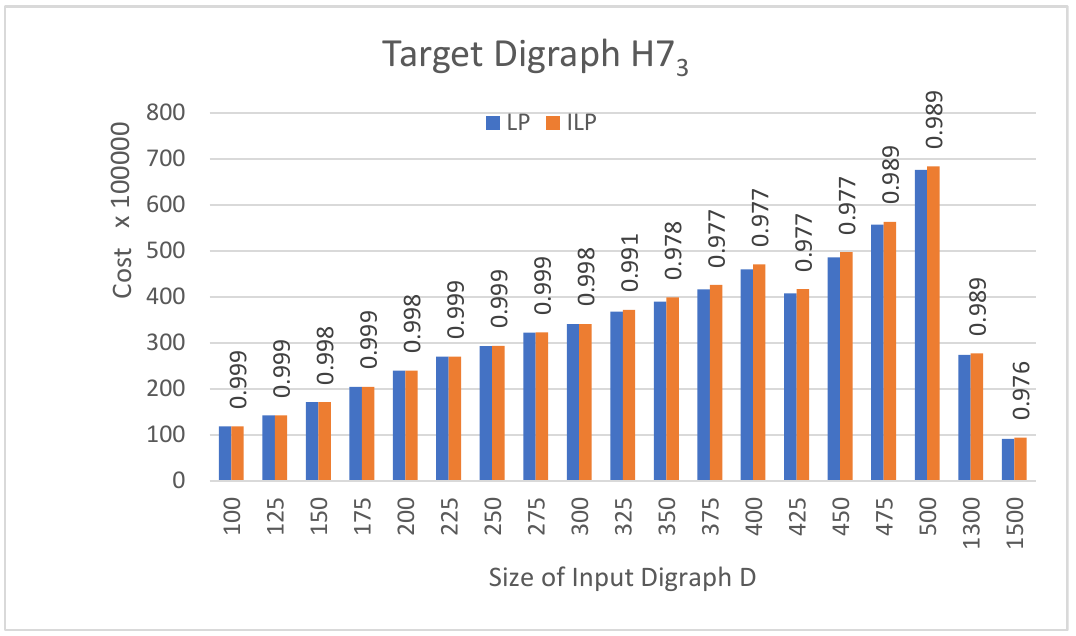}
  \end{subfigure}
  \hspace{2em}
  \begin{subfigure}{.2\linewidth}
    \includegraphics[scale=0.6]{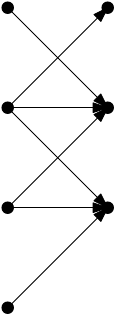}
   
  \end{subfigure}
  \caption{Comparison between LP and ILP when the target digraph is $H7_3$. $H7_3$ is a bipartite graph that admits a min-ordering.}
  \label{examples-H73}
\end{figure}

\begin{figure}[H]
  \centering
  \begin{subfigure}{.6\linewidth}
    \includegraphics[width =
    \linewidth]{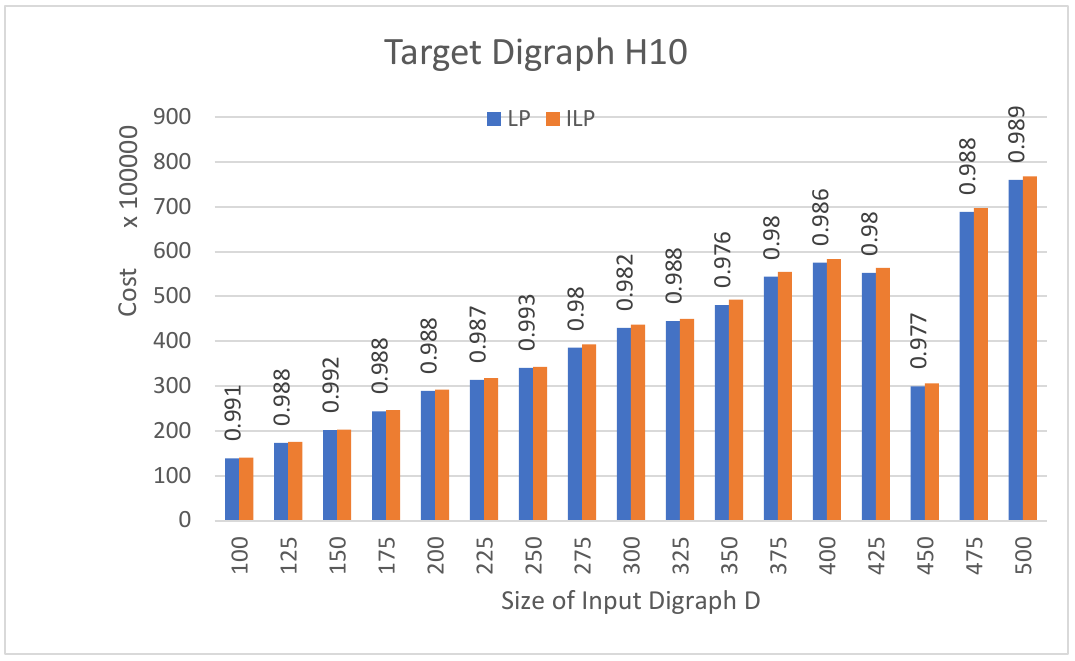}
   
  \end{subfigure}
  \hspace{2em}
  \begin{subfigure}{.2\linewidth}
    \includegraphics[scale=0.6]{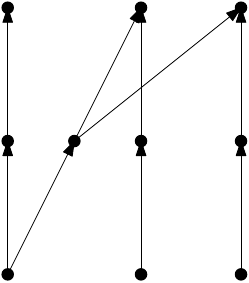}
  
  \end{subfigure}
  \caption{Comparison between LP and ILP when the target digraph is $H10$. $H10$ admits a min-ordering and it does not admit a majority polymorphism.}
  \label{examples-H10}
\end{figure}

\begin{figure}[ht!]
  \centering
  \begin{subfigure}{.6\linewidth}
    \includegraphics[width =
    \linewidth]{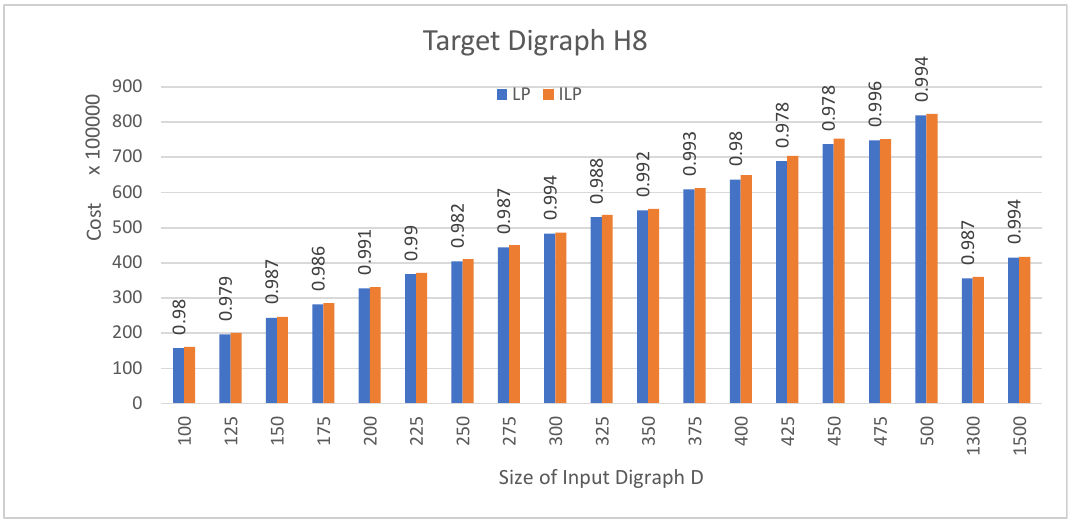}
  
  \end{subfigure}
  \hspace{2em}
  \begin{subfigure}{.2\linewidth}
    \includegraphics[scale=0.6]{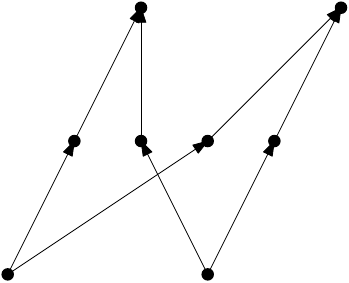}
  
  \end{subfigure}
  \caption{Comparison between LP and ILP when the target digraph is $H8$. $H8$ admits a conservative majority polymorphism. It does not admit a min-ordering.}
  \label{examples-min-ordering}
\end{figure}

\begin{figure}[H]
  \centering
  \begin{subfigure}{.6\linewidth}
    \includegraphics[width =
    \linewidth]{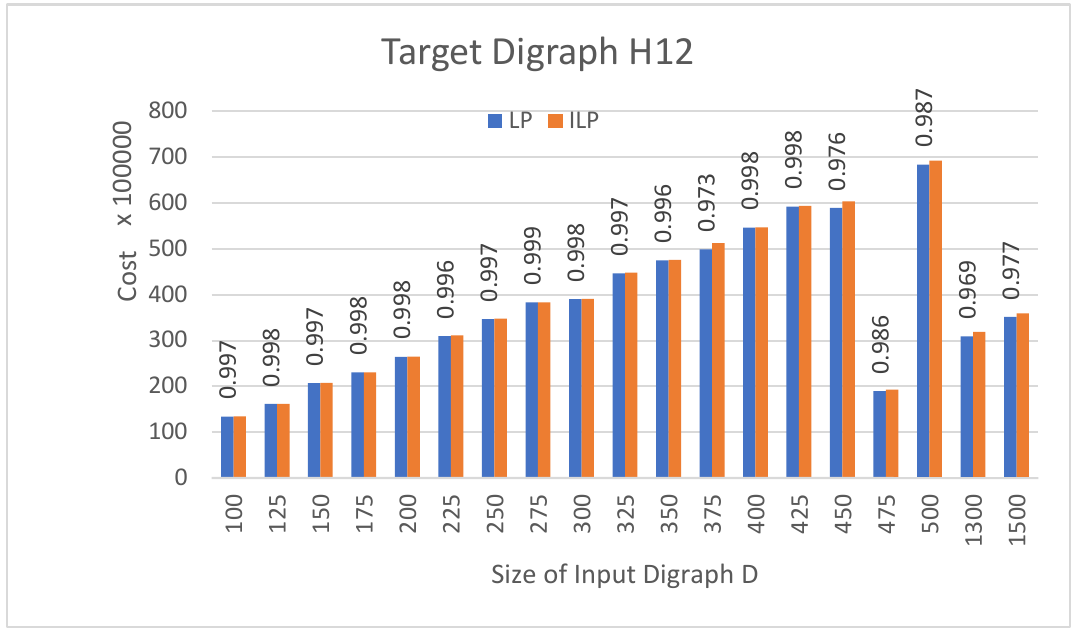}
  \end{subfigure}
  \hspace{1em}
  \begin{subfigure}{.2\linewidth}
    \includegraphics[scale=0.5]{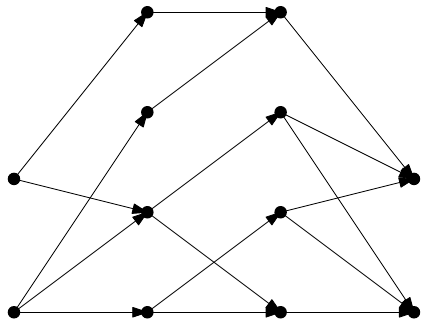}
  \end{subfigure}
  \caption{Comparison between LP and ILP when the target digraph is $H12$. $H12$ is a \emph{balanced} digraph that admits a min-ordering.}
  \label{examples-H12}
\end{figure}

\begin{figure}[H]
  \centering
  \begin{subfigure}{.6\linewidth}
    \includegraphics[width =
    \linewidth]{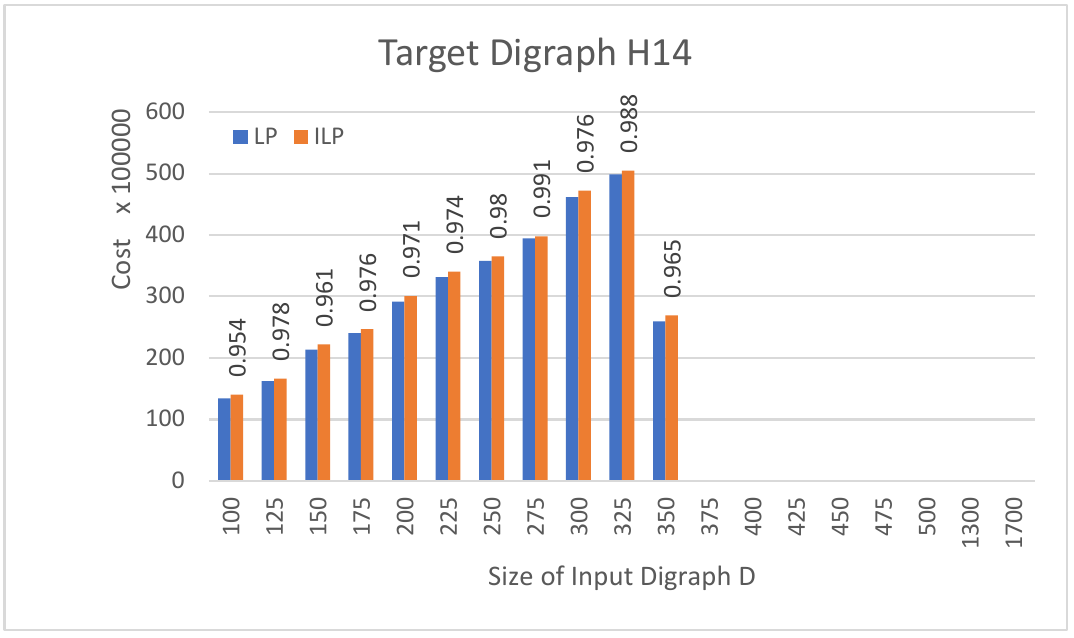}
  \end{subfigure}
  \hspace{2em}
  \begin{subfigure}{.2\linewidth}
    \includegraphics[scale=0.5]{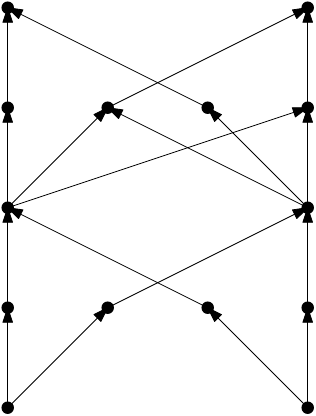}
  \end{subfigure}
  \caption{Comparison between LP and ILP when the target digraph is $H14$. $H14$ is a DAT-free digraph.}
  \label{examples-H14}
\end{figure}

\begin{figure}[H]
  \centering
  \begin{subfigure}{.6\linewidth}
    \includegraphics[width =
    \linewidth]{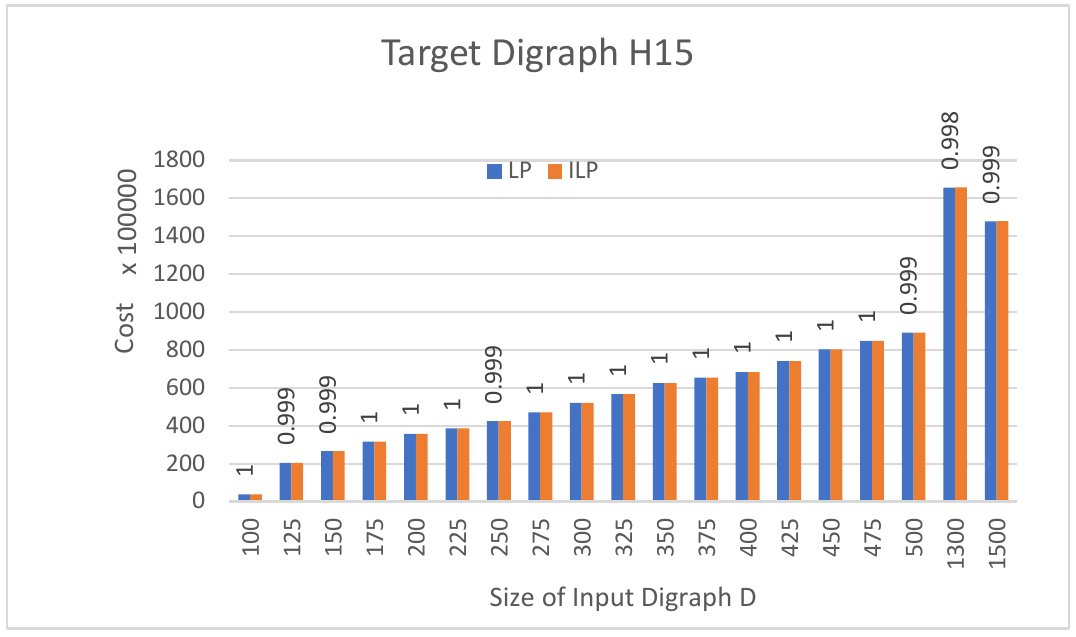}
  \end{subfigure}
  \hspace{1em}
  \begin{subfigure}{.2\linewidth}
    \includegraphics[scale=0.5]{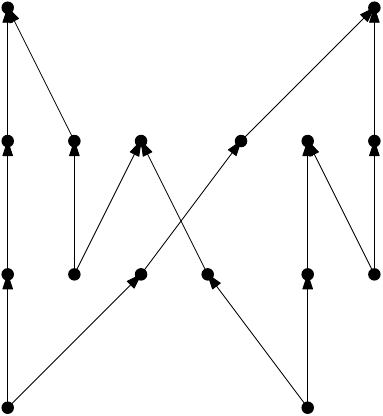}
   
  \end{subfigure}
  \caption{Comparison between LP and ILP when the target digraph is $H15$. $H15$ admits a conservative majority polymorphism. It does not admit a min-ordering.}
  \label{examples-H16}
\end{figure}

\section{Conclusion} In this paper we study the approximation of MinHOM problem. We present several positive results proving that for digraphs $H$ which admit a min-ordering or a $k$-min-ordering MinHOM($H$) admits a constant factor approximation algorithm. Moreover, we obtain a complete classification of graphs for which MinHOM is approximable within a constant factor. 

We complement our theoretical results with an empirical study of the performance of our algorithm. We have implemented and run our algorithm on several examples of bipartite graphs and digraphs with min-ordering as well as on digraphs without min-ordering. The weights have been randomly chosen from a much larger range than the number of vertices of input digraph $G$. The implementation of our algorithms provides 
a much better approximation ratio than our  theoretical bounds. 
It leaves open to investigate a classification of digraphs $H$, where MinHOM($H$) admits a constant factor approximation algorithm
that is independent of $|V(H)|$.

\noindent \textbf{Acknowledgement:} We are thankful to Andrei Bulatov for proofreading several drafts of the work and many valuable discussions that significantly improved the paper and its presentation.


\bibliography{main.bbl}

\end{document}